\documentclass[a4paper,cleveref, pdfa, UKenglish]{lipics-v2021}

\usepackage{stmaryrd}
\usepackage[dvipsnames,table]{xcolor}
\newcolumntype{x}[1]{>{\centering\arraybackslash\hspace{0pt}}p{#1}}

\usepackage{tikz}
\usetikzlibrary{decorations.pathmorphing,3d}

\newcommand{\Z}{\mathbb{Z}}
\newcommand{\N}{\mathbb{N}}
\DeclareMathOperator{\dom}{dom}

\DeclareMathOperator\level{level}

\newcommand{\uo}{{\pmb\rightarrowtriangle}} 
\newcommand{\uc}{{\pmb\leftarrowtriangle}}
\newcommand{\uor}{\gray\rightarrow}
\newcommand{\ucr}{\gray\leftarrow}

\newcommand{\ulc}{\ulcorner}
\newcommand{\urc}{\urcorner}
\newcommand{\llc}{\llcorner}
\newcommand{\lrc}{\lrcorner}
\newcommand{\sh}{\,|\,}
\newcommand{\sv}{\textbf{---}}

\newcommand{\gray}[1]{\textcolor{lipicsBulletGray}{#1}}
\newcommand{\noop}[1]{\pmb{#1}}

\newcommand{\dalphabet}{\Sigma \times \{\uo,\uc\}}

\newcommand{\DecisionProblem}[2]{\begin{description}
\item [Input:] #1
\item [Output:] #2
\end{description}}

\hideLIPIcs
\nolinenumbers

    \title{The aperiodic Domino problem in higher dimension}

    \author{Antonin {Callard}}{Université Paris-Saclay,  ENS Paris-Saclay, Département Informatique, 91190 Gif-sur-Yvette, France \and \url{https://www.acallard.net/} }{contact@acallard.net}{https://orcid.org/0000-0002-4673-4881}{}
    
    \author{Benjamin {Hellouin de Menibus}\footnote{Corresponding author}}{Université Paris-Saclay, CNRS, Laboratoire Interdisciplinaire des Sciences du Numérique, 91400 Orsay, France \and \url{https://www.lisn.upsaclay.fr/~hellouin/}}{hellouin@lisn.fr}{https://orcid.org/0000-0001-5194-929X}{}

\authorrunning{A. Callard and B. Hellouin de Menibus}

\Copyright{Antonin Callard and Benjamin Hellouin de Menibus}

\ccsdesc[500]{Theory of computation~Models of computation}
\ccsdesc[300]{Theory of computation~Problems, reductions and completeness}

\keywords{Subshift, periodicity, aperiodicity, domino problem, subshift of finite type, sofic subshift, effective subshift, tilings, computability}

\acknowledgements{The authors are grateful to the three referees for their many remarks and improvements. The first author is grateful to the \emph{Excellence} group for being such cheerful co-interns and friends.}

\EventEditors{Petra Berenbrink and Benjamin Monmege}
\EventNoEds{2}
\EventLongTitle{39th International Symposium on Theoretical Aspects of Computer Science (STACS 2022)}
\EventShortTitle{STACS 2022}
\EventAcronym{STACS}
\EventYear{2022}
\EventDate{March 15--18, 2022}
\EventLocation{Marseille, France}
\EventLogo{}
\SeriesVolume{219}
\ArticleNo{18}

\begin{document}
\maketitle

\begin{abstract}
The classical Domino problem asks whether there exists a tiling in which none of the forbidden patterns given as input appear. In this paper, we consider the aperiodic version of the Domino problem: given as input a family of forbidden patterns, does it allow an aperiodic tiling? The input may correspond to a subshift of finite type, a sofic subshift or an effective subshift.

\cite{2018-GHV} proved that this problem is co-recursively enumerable ($\Pi_0^1$-complete) in dimension 2 for geometrical reasons. We show that it is much harder, namely analytic ($\Sigma_1^1$-complete), in higher dimension: $d \geq 4$ in the finite type case, $d \geq 3$ for sofic and effective subshifts. The reduction uses a subshift embedding universal computation and two additional dimensions to control periodicity.

This complexity jump is surprising for two reasons: first, it separates 2- and 3-dimensional subshifts, whereas most subshift properties are the same in dimension 2 and higher; second, it is unexpectedly large. 

\end{abstract}
\newpage

\section{Introduction}

Subshifts are sets of colorings (or \emph{configurations}) defined by a family of forbidden patterns. The seminal computational problem on multidimensional subshifts is the \emph{Domino Problem}: given a subshift of finite type (SFT), does it contain a configuration? It was proved undecidable on $\Z^2$ in \cite{1966-Berger,1971-Robinson} from the construction of \emph{aperiodic SFTs} (SFTs which contain only (strongly) aperiodic colorings, i.e. colorings with no non-zero period) in which universal computation is embedded. Many similar undecidability results used different SFTs and embeddings to control the structure and properties of their configurations \cite{2008-Hochman,2010-DRS,2013-AS,2017-Westrick,2010-DRS} or to characterise the set of possible values of some parameters by computability conditions \cite{2010-HM,2011-Meyerovitch,2021-CV}. These results all rely on the existence of purely aperiodic SFTs on $\Z^d$ for $d \geq 2$ (see \cite[Section 1.2]{2020-UndecidabilityDomino} for more details), and show that multidimensional SFTs can be considered as geometrical computational models.

In contrast, topological or geometrical restrictions may lower the ``natural'' complexity of a problem (compare e.g. \cite{2010-HM} with \cite{1997-Friedland} or \cite{2015-PS}) by breaking our ability to embed computation. In particular, finding the border where the difficulty jump occurs gives a fine understanding of the effect of the restriction \cite{2019-GH}.\medskip

Given the importance of aperiodicity for computation embedding, it is natural to ask the aperiodic version of the Domino problem (\textbf{AD}): \emph{given as input a subshift $X$, does it contain an aperiodic coloring?} It is not difficult to see that this problem is harder than the Domino problem, i.e. co-r.e.($\Pi_1^0$)-hard in dimension 2 and higher; however, the natural upper bound is much higher, outside the arithmetical hierarchy.

This question was (to the best of our knowledge) first explored in \cite{2018-GHV}: the authors proved that \textbf{AD} is $\Pi_1^0$-complete for $\Z^2$ subshifts\footnote{on SFTs, but as \cite[Theorem~1]{2018-GHV} applies to any $\Z^2$ subshift, the result also holds for $\Z^2$ effective subshifts.}. It is an example of problem whose computational complexity is low because of geometrical reasons specific to the two-dimensional case: starting from an aperiodic configuration, we can regroup breaks of periods into concentric balls whose size is controlled by a computable function (\cite[Theorem 1]{2018-GHV}). \medskip

In this paper, we study the computational complexity of this problem in higher dimension, where this geometrical property no longer holds (see \cite[Section 4]{2018-GHV} for a counter example). We build an embedding for universal computation that proves that this problem is in a much higher undecidability class -- $\Sigma_1^1$-complete, its natural upper bound -- in sofic subshifts for $d\geq 3$ and in subshifts of finite type for $d\geq 4$.

Our paper is structured as follows.
\begin{itemize}
\item In \Cref{sec:definitions-notations}, we provide definitions for subshifts and the relevant complexity classes;
\item In \Cref{sec:sofic-highly-undecidable}, we prove that \textbf{AD} is $\Sigma_1^1$-complete on $\Z^d$ ($d \geq 3$) sofic subshifts;
\item In \Cref{sec:sft-highly-undecidable}, we adapt the previous proof to $\Z^d$ ($d \geq 4$) subshifts of finite type;
\item In \Cref{sec:low-complexity}, we make a side remark relating the existence of an aperiodic configuration in SFTs with their complexity.
\end{itemize}

We summarize the complexity of \textbf{AD} in the following table (new results are highlighted):

\begin{center}
  \begin{tabular}{c|x{2cm}|x{2cm}|x{2cm}}
    Dimension / type & finite type & sofic & effective\\
    \hline
    2D$\;$ & $\Pi_1^0$-complete& $\Pi_1^0$-complete& $\Pi_1^0$-complete\\
    3D$\;$ & open & \cellcolor{blue!15!white}$\Sigma_1^1$-complete & \cellcolor{blue!15!white}$\Sigma_1^1$-complete\\
    $\;$4D+ & \cellcolor{blue!15!white}$\Sigma_1^1$-complete & \cellcolor{blue!15!white}$\Sigma_1^1$-complete & \cellcolor{blue!15!white}$\Sigma_1^1$-complete \\
  \end{tabular}
\end{center}

Considering the effect of the dimension on the difficulty of \textbf{AD}, we find a border between the dimensions where the complexity of the problem is lowered by geometric properties and the dimensions where computability considerations dominate. 

For sofic and effective subshifts, this border lies between dimensions 2 and 3. For SFTs on $\Z^3$, we conjecture that \textbf{AD} is in the arithmetical hierarchy for reasons that are specific to SFTs, and provide a few pointers in conclusion. This would be a candidate for a dimension-separating property between 3 and 4 dimensional SFTs. In both cases, we do not know of any other natural problem with a complexity jump in such high dimensions.

\section{Definitions and notations}\label{sec:definitions-notations}
\subsection{Subshifts}

For a more detailed introduction, we refer the reader to \cite[Chapter 9]{2016-CombinatoricsWordsandSD}.

Let $\Sigma$ be a finite alphabet of colors and $d$ a dimension. A \emph{configuration} is a coloring $c : \Z^d \mapsto \Sigma$, and the value of $c$ at position $i$ is denoted $c_i$. A \emph{pattern} is a coloring $w : D \mapsto \Sigma$ of a finite domain $D = \dom(w) \subseteq \Z^2$. We say that a pattern~$w$ appears in a configuration $x$ and write $w \sqsubseteq x$ if $w_j = x_{i+j}$ for some $i \in \Z^d$ and all $j \in D$. Given a configuration $x$ and a vector $t \in \Z^d$, denote $\sigma^t(x)$ the shift of $x$ by $t$: for any $i \in \Z^d, \sigma^t(x)_i = x_{i-t}$.

\begin{definition}[Periodicity]\label{period-subshifts}
  \begin{enumerate}
  \item In a configuration $x \in \Sigma^{\Z^d}$, a vector $p \in \Z^d$ is \emph{broken at position $i$} if $x_{i+p} \neq x_{i}$.
  \item A configuration $x \in \Sigma^{\Z^d}$ is (strongly) \emph{aperiodic} if every vector $p \in \Z^d$ is broken in $x$.
  \end{enumerate}
\end{definition}

In the following definition, $\Sigma$ is equipped with the discrete topology and $\Sigma^{\Z^d}$ with the product topology. $\Sigma^{\Z^d}$ is then a Cantor space.

\begin{definition}[Subshifts]\label{def:subshifts}
  A \emph{subshift} is a closed and $\sigma$-invariant subset of $\mathbb Z^d$. Equivalently, there is a family of forbidden patterns~$\mathcal{F}$ such that
    \[ X = X_{\mathcal{F}} := \left \{ x \in \Sigma^{\Z^d} : \forall w \in \mathcal{F}, w \not \sqsubseteq x \right \}\]
\end{definition}

Two distinct families of forbidden patterns may define the same subshift. 
\begin{definition}[Classes of subshifts]\label{def:subshifts-classes}
A subshift $Y \subseteq \Sigma^{\Z^d}$ is:
  \begin{enumerate}
  \item \emph{of finite type} (\emph{SFT}) if it can be defined by a finite family of forbidden patterns.
  \item \emph{sofic} if there exists an SFT $X \subseteq {\Sigma'}^{\Z^d}$ and a projection $\pi : \Sigma' \mapsto \Sigma$  such that $Y = \pi(X)$.
  \item \emph{effective} if it can be defined by a recursively enumerable family of forbidden patterns.
  \end{enumerate}
\end{definition}

SFTs are of course sofic and sofic subshifts are effective. On the other direction, effective subshifts are projections of higher-dimensional sofic subshifts; this is a consequence of \cite{2008-Hochman}, later improved in the subshift case in \cite{2010-DRS,2013-AS}. More precisely, $X^\uparrow \subseteq \Sigma^{\Z^{d+k}}$ is a $(d+k)$-dimensional \emph{lift} of a subshift $X \subseteq \Sigma^{\Z^d}$ if its configurations are configurations of $X$ repeated along the $k$ additional dimensions. Then:

\begin{theorem}[\cite{2010-DRS,2013-AS}]\label{thold:effective-sofic}
  For any $\Z^d$ effective subshift $X$, its $(d+1)$-dimensional lifts are sofic.
\end{theorem}

\subsection{Hierarchy of undecidability}

Many-one reductions define a preorder on decision problems (``$P_1$ is easier than $P_2$''), so we can define hierarchies according to ``how far'' a problem is from being computable. 

\paragraph*{Arithmetical hierarchy}

Starting from recursively enumerable ($\Sigma_1^0$) and co-recursively enumerable ($\Pi_1^0$) problems, the arithmetical hierarchy progressively defines higher levels of undecidability.

\begin{definition}[Arithmetical hierarchy]\label{def:arithm-hierarchy}
  For a decision problem $P: \N\mapsto\{0,1\}$ and $m\geq 1$,
  \begin{enumerate}
  \item $P \in \Sigma_m^0$ if there is a computable relation $R(n,k_1,...,k_m)$ such that
    \[ P(n) = 1 \Leftrightarrow \exists k_1, \forall k_2,\exists k_3,\dots R(n,k_1, \dots,k_m).\]
  \item $P \in \Pi_m^0$ if this definition holds when swapping $\forall$ and $\exists$ quantifiers.
  \end{enumerate}

$P$ is \emph{arithmetical} if it belongs to a level of this hierarchy.
\end{definition}

As $\Sigma_m^0 \cup \Pi_m^0 \subseteq \Sigma_{m+1}^0 \cap \Pi_{m+1}^0$, this indeed defines a hierarchy. For more details, we refer the reader to \cite[Chapter 4]{2016-TuringComputability}.

\paragraph*{Analytical hierarchy}

Above the arithmetical hierarchy, the analytical hierarchy allows for second-order quantifications on sets. Here we need only the first level.

\begin{definition}[Class $\Sigma_1^1$]\label{def:sigma-1-1}
  A decision problem $P : \N\mapsto \{0,1\}$ is $\Sigma_1^1$ if there exists an \emph{arithmetical} relation $R$ such that
  \[ P(n) = 1 \Leftrightarrow \exists f \in 2^\N, R^f(n) \]
  in which $R^f$ denotes the relation $R$ with $f$ given as an oracle.
\end{definition}

All arithmetical sets are $\Sigma_1^1$. In terms of computational power, $\Sigma_1^1$ sets are (a lot) harder than arithmetical sets: to make an analogy between computability and topology, if $\Sigma_1^0$ sets correspond to the open sets, then $\Sigma_1^1$ sets are not even Borel. For more details, see \cite[Chapter IV.2]{1989-ClassicalRecursionTheory}. A typical example of a \emph{$\Sigma_1^1$-complete problem} is the following:

\begin{theorem}[State Recurrence {\cite[Corollary 6.2]{1986-Harel}}]\label{thold:state-recurrence}
  The problem of \textbf{State Recurrence (SR)}:\DecisionProblem{A nondeterministic Turing machine (NTM) $\mathcal{M}$, and one of its states $q_0$.}{Is there a run of $\mathcal{M}$ on the empty input $\varepsilon$ in which $q_0$ is visited infinitely often?}is a $\Sigma_1^1$-complete problem.
\end{theorem}

\subsection{The aperiodic Domino (\textbf{AD}) problem and its complexity}\label{sec:general-problem}
\begin{definition}[Aperiodic Domino problem (\textbf{AD})]\label{def:eac}~
  \DecisionProblem{An effective family of $d$-dimensional patterns.}{Is there an aperiodic configuration in the effective subshift $X_\mathcal{F}$?}
\end{definition}
  We consider variations of \textbf{AD} depending on the type of input subshift  (SFT, sofic, effective). There are natural lower and upper bounds on the complexity of \textbf{AD} that do not depend on the input type:

\begin{proposition}\label{lemma:pi-hard}
  \textbf{AD} is $\Pi_1^0$-hard for $\Z^d$ subshifts ($d \geq 2$).
\end{proposition}

\begin{proof}
  We reduce the Domino problem to \textbf{AD}. Let $Y$ be a $\Z^d$-SFT with only aperiodic configurations (see e.g. \cite{1971-Robinson}).
For any $\Z^d$ subshift $X$, the cartesian product $X \times Y$ has the same type (SFT, sofic, effective), has only aperiodic configurations, and is non-empty if and only if $X$ is non-empty.
\end{proof}

\begin{proposition}\label{th:existence-of-an-aperiodic-configuration-sigma-1-1}
  \textbf{AD} is a $\Sigma_1^1$ problem for $\Z^d$ subshifts.
\end{proposition}

\begin{proof}
  Let $\mathcal{F}$ be the effective family of forbidden patterns given as input. The existence of an aperiodic configuration can be written as:
  \[\exists x \in \Sigma^{\Z^d}, x \in X_{\mathcal{F}} \text{ and } x \text{ is aperiodic.}\]
  Taking any computable encoding between $\Sigma^{\Z^d}$ and $2^\N$, we can see that the first (existential) quantifier is of second order and can be written as a quantifier on $2^\N$. The rest of the expression is a $\Pi_2^0$ relation, and in particular arithmetical ($x$ being given as oracle):
  \begin{itemize}
    \item $x \in X_\mathcal{F} \Leftrightarrow \forall w \in \mathcal{F}, \forall i \in \Z^d, x_{|i + \dom(w)} \neq w$;
    \item $x$ is aperiodic $\Leftrightarrow\forall p \in \Z^d, \exists i \in \Z^d, x_{i} \neq x_{i+p}$.\qedhere
  \end{itemize}
\end{proof}

\section{\texorpdfstring{$\Sigma_1^1$}{Sigma\_{}1\^{}1}-completeness for \texorpdfstring{$\Z^d$}{Z\^{}d} sofic and effective subshifts, \texorpdfstring{$d\geq 3$}{d>=3}}
\label{sec:sofic-highly-undecidable}

\begin{theorem}\label{th:3d-sofic}
\textbf{AD} for $\Z^d$ sofic subshifts, $d\geq 3$, is a $\Sigma_1^1$-complete problem.
\end{theorem}

By \cref{th:existence-of-an-aperiodic-configuration-sigma-1-1}, $\textbf{AD} \in \Sigma_1^1$. We prove $\Sigma_1^1$-hardness for $d=3$ and the higher-dimensional cases will follow. 

To prove $\Sigma_1^1$-hardness, we reduce (many-one reduction) the problem \textbf{SR}. Let $\mathcal{M}$ be some nondeterministic Turing machine (NTM) and $q_0$ one of its states. We create a sofic subshift $Y_3$ which contains an aperiodic configuration if and only if $\mathcal{M}$ admits a run from the empty word which visits $q_0$ infinitely often. The proof is divided in three parts:
\begin{enumerate}
\item \Cref{sec:1d-toeplitz}: creation of an auxiliary $\Z$ Toeplitz subshift $T_\mathcal{M}$;
\item \Cref{sec:sofic-y3}: creation of $Y_3$ and proof that it is sofic;
\item \Cref{sec:3d-reduction}: proof that $Y_3 \in \textbf{AD}$ iff $(\mathcal{M},q_0) \in \textbf{SR}$.
\end{enumerate}

\subsection{\texorpdfstring{$T_\mathcal{M}$}{T\_M}: \texorpdfstring{$\Z$}{Z} Toeplitz corresponding to state sequences of \texorpdfstring{$\mathcal{M}$}{M}}\label{sec:1d-toeplitz}

In this section, we transform the set of sequences of states in all the runs of $\mathcal{M}$ into a $\Z$ subshift with a convenient structure called \emph{Toeplitz}.

\paragraph*{$\Z$ Binary Toeplitz subshift $X_T$}

Consider the substitution $\sigma$ on the alphabet $\{\uor,\ucr,\uo,\uc\}$:
\[  \sigma(\uor) = \uo \ucr \qquad \sigma(\ucr)= \uc \ucr\qquad \sigma(\uo) = \uo \uor \qquad \sigma(\uc) = \uc \uor \]
Define the $\Z$-subshift $X_\sigma$ by forbidding every pattern that does not appear in $\sigma^\omega(\uo)$ (any other seed symbol would yield the same subshift).
\[ \sigma^\omega(\uo) = \uo \uor \uo \ucr \uo \uor \uc \ucr \uo \uor \uo \ucr \uc \uor \uc \ucr \, \dots \]

\begin{definition}[Binary Toeplitz subshift]\label{def:binary-toeplitz}
  The \emph{binary Toeplitz subshift} $X_T$ is the image of $X_\sigma$ under the projection that maps $\{\uor, \uo\}$ to $\uo$ and $\{\ucr, \uc\}$ to $\uc$.
\end{definition}

$X_T$ is a Toeplitz subshift \cite{1969-JK}. It corresponds to the "period-doubling" or "ruler" (modulo 2) sequences (resp. A001511 and A096268 in the OEIS). In a configuration of~$X_T$,
\begin{description}
\item[Level 1] One position out of two has an alternating sequence of $\uo$ and $\uc$;
\item[Level 2] One position out of two \emph{in the remaining positions} (i.e. one out of four) has an alternating sequence of $\uo$ and $\uc$, etc.
\end{description}

More generally, a position $i$ is of level $\ell$ if it has minimal period $2^{\ell+1}$ (cells at positions $i+k2^{\ell+1}$ all have the same value). In a configuration of $X_T$, there may exist at most one position $i$ which does not have a level (i.e. it is not $2^{\ell+1}$ periodic for any $\ell$): we say that $\level_z(i)=\infty$. Given as input a finite pattern of size between $2^n$ and $2^{n+1}$, one can compute all levels $\ell \leq n-1$.

\paragraph*{$T_\mathcal{M}$ : $\Z$-Toeplitzification of sequences of states of $\mathcal{M}$}

\begin{definition}[Toeplitzification of a set of sequences]\label{def:toeplitzification}
Given a set of sequences $A \subseteq \Sigma^\N$, we define the corresponding Toeplitzified subshift $T_A$ on the alphabet $\Sigma \times \{\uo,\uc\}$ as:
\[ T_A = \left\{ (x,z) \in (\Sigma \times \{\uo,\uc\})^\Z :
  \begin{array}{l}
    z \in X_T, \exists (a_n)_{n \in \N} \in A, \\
    \forall i \in \Z, \level_z(i) = \ell \in \N \implies x_i = a_\ell
  \end{array}\right\} \]
\end{definition}

Note that a position of infinite level may be marked with any symbol of $\Sigma$. We cannot force this symbol without breaking the next lemma.

Now take $\Sigma = Q$, the set of states of $\mathcal{M}$, and define $S_\mathcal{M}$ as the set of sequences $(s_t)_{t \in \mathbb{N}}$ on the alphabet $\Sigma$ such that there exists a non-terminating run of $\mathcal{M}$ from the empty input whose state at time $t$ is~$s_t$. Let $T_\mathcal{M} := T_{S_\mathcal{M}}$ be its Toeplitzification.

\begin{lemma}\label{lemma:tm-effective}
$T_\mathcal{M}$ is a $\Z$ effective subshift.
\end{lemma}

\begin{proof}
This stems from the fact that the set of prefixes of $S_\mathcal{M}$ is computable: for any $n\geq 0$, we can enumerate all oracles of non-determinism of $\mathcal{M}$ of size~$\leq n$ and compute $S_n$, the set of finite prefixes of length $n$ in $S_\mathcal{M}$.

Consider the following algorithm that defines a family of forbidden patterns. For all $n$:
\begin{itemize}
\item Compute the globally admissible patterns of $X_T$ of size $2^n+1$; (\emph{Note that the language of patterns of $X_T$ is computable: it is both recursively and co-recursively enumerable.)}
\item Compute $\mathcal S_n$ ;
\item Forbid all patterns $(u,v) \in {\left(\dalphabet\right)}^{2^n+1}$, except if $v$ is a pattern in $X_T$ and there exists a prefix $(s_t)_{0 \leq t \leq n} \in \mathcal S_n$ such that: 
    \[\forall i,j\in\Z,\; \left(\level_{v}(i) = \level_{v}(j) \leq n\right) \implies u_i = u_j = s_{\level_{v}(i)}.\]
\end{itemize}
This procedure defines an effective subshift $E$. We prove $E = T_\mathcal{M}$. Indeed:
\begin{description}
\item [$E \subseteq T_{\mathcal M}$] Take $(u, v) \in E$ and $(u^n,v^n) = (x,z)[-2^n,2^n]$. By definition of $E$, there exists a finite prefix $s^n\in S_n$ such that for any positions $i,j$ with $\level_{v^n}(i) = \level_{v^n}(j) = \ell$, we have $u^n_i = u^n_j = s_\ell$. This sequence of prefixes is increasing, so it converges towards some sequence $s \in S_\mathcal{M}$. Then for any $i,j \in \mathbb{Z}$ such that $\level_{v}(i) = \level_{v}(j) = \ell < +\infty$, we have $x_i = x_j = s_\ell$. So $(x,z) \in T_\mathcal{M}$.
\item [$T_{\mathcal M} \subseteq E$] No pattern forbidden in the algorithm appears in any configuration of $T_\mathcal{M}$. \qedhere
\end{description}
\end{proof}

\subsection{\texorpdfstring{$Y_3$}{Y\_3}: the desired \texorpdfstring{$\Z^3$}{Z\^{}3} subshift}\label{sec:sofic-y3}

We create a subshift $Y_3$ which contains an aperiodic configuration if and only if there exists a run of $\mathcal{M}$ on the empty word which visits $q_0$ infinitely often. As one might expect, each configuration of $Y_3$ contains the lift of a configuration of $T_\mathcal{M}$ corresponding to a run of $\mathcal{M}$. We then add lines to make it aperiodic if and only if $q_0$ appears infinitely often.

However, every decision of breaking periods must occur locally at every level, without the ability to know whether the future number of visits of $q_0$ is finite or infinite. Otherwise compactness would create issues: as visits of $q_0$ can occur arbitrarily late, a position of finite level could be tricked to ``believe'' that $q_0$ is visited infinitely often in the future. That is why we will break periods whose size depend on the level of the positions in the Toeplitz structure.

\paragraph*{Effective 2D subshifts: $Y_2^{\uo}$ and $Y_2^{\uc}$}
A configuration of $Y_2^{\uo}$ is composed of three layers:
\begin{itemize}
\item Layer 1 \& 2 : it contains a $\Z^2$ lift of a configuration $x' \in T_\mathcal{M}$. That is, $\forall i,j : x_{i,j} = x'_{i}$.
\item Layer 3: on the alphabet $\{\blacksquare,\square\}$. For every $\ell$ and in every column of level $\ell$ containing $(q_0, \uo)$ on Layers 1 and 2, Layer 3 contains regularly placed $\blacksquare$ cells separated by $2^\ell-1$ $\square$ cells. Every other cell contains $\square$ on Layer 3.
\end{itemize}

Formally, $Y_2^{\uo}$ can be written as:
\[
    \left\{ \begin{array}{l}x \in \,{\left(\dalphabet \times \{\blacksquare,\square \} \right)}^{\Z^2} : \exists x' \in T_\mathcal{M}, \pi_{1,2}(x) \text{ is a $\Z^2$ lift of $x'$}, \\
  \hspace{4.96cm} \exists z \in X_T, \pi_2(x) \text{ is a $\Z^2$ lift of $z$},\\
  \forall i,\,j \in \Z,\\
   \qquad  x_{i,j} = (\cdot,\cdot,\blacksquare) \implies \forall j', x_{i,j'} = (q_0,\uo,\cdot) \\
   \qquad x_{i,j} = (q_0,\uo,\cdot) \text{ and } \level_z(i) = \ell \in \N \implies \exists j', \forall j'', x_{i,j''} = (\cdot, \cdot,\blacksquare) \iff 2^\ell \mid (j'' - j') \\
  \qquad x_{i,j} = (q_0,\uo,\cdot) \text{ and } \level_z(i) = \infty \implies |\{ j' : x_{i,j'} = (\cdot, \cdot,\blacksquare)\}| \leq 1
    \end{array}\right\}
 \]

Its counterpart $Y_2^{\uc}$ is defined similarly by replacing $\uo$ by $\uc$ in the previous definition. 
It is clear that both $Y_2^{\uo}$ and $Y_2^{\uc}$ are effective $\Z^2$ subshifts.

\paragraph*{Issues with the position of infinite level}\hypertarget{par:infinite-level-aper}{}

Note that $\blacksquare$ symbols break increasingly large periods as levels in the Toeplitz structure increase: by compactness, a position of infinite level can break periods of every size by~itself.

This explains why this construction requires two additional dimensions to $T_\mathcal{M}$ instead of one: each position in $Y_3$ will be periodic in one dimension, and breaks periods in the other. This way, the single position of infinite level may break horizontal or vertical periods, but not both.

\paragraph*{Sofic 3D subshifts: $Y_3^{\uo}$ and $Y_3^{\uc}$}
By \Cref{thold:effective-sofic}, every $d$-dimensional effective subshift can be lifted into a $d+1$-dimensional sofic subshift. Using this result, we lift $Y_2^{\uo}$ and $Y_2^{\uc}$ into 3D sofic subshifts $Y_3^{\uo}$ and $Y_3^{\uc}$:
\begin{align*}
  Y_3^{\uo} & = \{ x \in {\left(\dalphabet \times \{\blacksquare,\square \} \right)}^{\Z^3} : \exists x' \in Y_2^{\uo}, \forall i,k \in \Z, \forall j' \in \Z, x_{i,j',k} = x'_{i,k} \} \\
  Y_3^{\uc} & =  \{ x \in {\left(\dalphabet \times \{\blacksquare,\square \} \right)}^{\Z^3} : \exists x' \in Y_2^{\uc}, \forall i,j \in \Z, \forall k' \in \Z, x_{i,j,k'} = x'_{i,j} \}
\end{align*}
\emph{Note that the lifts are not made along the same coordinates: a position with $\blacksquare$ in $Y_2^{\uo}$ lifts into a line directed by $(0,1,0)$ in $Y_3^{\uo}$, and a position with $\blacksquare$ in $Y_2^{\uc}$ lifts into a line directed by $(0,0,1)$ in $Y_3^{\uc}$.}
\paragraph*{Sofic 3D subshift: $Y_3$}

We obtain $Y_3$ by "fusing" the two previous subshifts. Formally,
   \[Y_3 = \{ x \in \, {\left(\dalphabet \times \{\blacksquare,\square\} \times \{\blacksquare,\square\}\right)}^{\Z^3} : \pi_{1,2,3}(x) \in Y_3^{\uo} \text{ and } \pi_{1,2,4}(x) \in Y_3^{\uc}\}.\]

Since $Y_3^{\uo}$ and $Y_3^{\uc}$ are sofic, their cartesian product $Y_3^{\uo} \times Y_3^{\uc}$ is also sofic. $Y_3$ is the projection on Layers $1,2,3,6$ of $Y_3^{\uo} \times Y_3^{\uc}$ with the additional local condition that the first two layers coincide (i.e. $\pi_{1,2}(x) = \pi_{4,5}(x)$), so it is sofic as well.

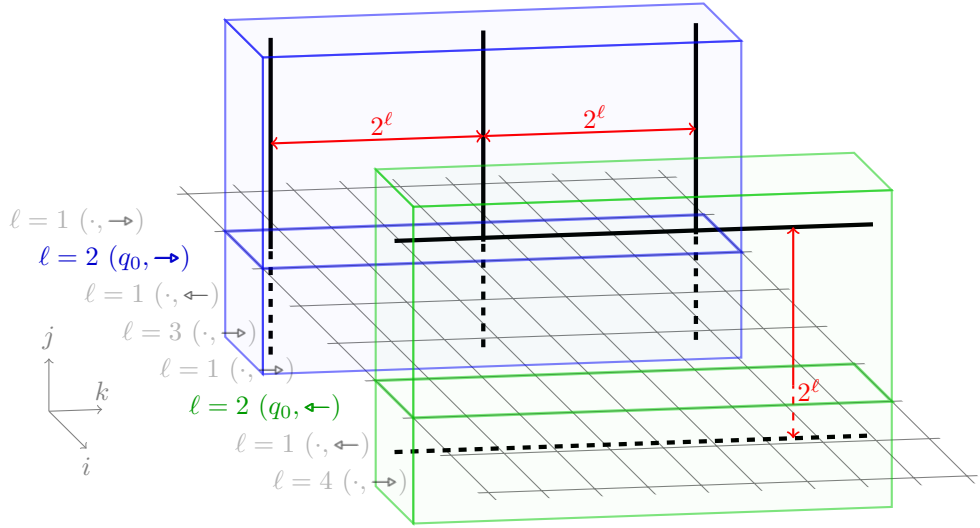
\begin{figure}[t]
  \centering
  \begin{tikzpicture}[scale = 0.7, x={(-45:1cm)},y={(182:1cm)},z={(90:1cm)}]
  \begin{scope}[canvas is xy plane at z=0]
    \clip (-0.2,0.2) rectangle (8.2,-9.2);
    \draw[gray] (-1,1) grid (9,-10);
  \end{scope}
  \begin{scope}[canvas is yz plane at x=1.5]
    \draw[ultra thick] (-0.5,0) -- (-0.5,4);
    \draw[ultra thick,dashed] (-0.5,-2) -- (-0.5,0);
    \draw[ultra thick] (-4.5,0) -- (-4.5,4);
    \draw[ultra thick,dashed] (-4.5,-2) -- (-4.5,0);
    \draw[ultra thick] (-8.5,0) -- (-8.5,4);
    \draw[ultra thick,dashed] (-8.5,-2) -- (-8.5,0);
    \draw[thick, red,<->] (-0.5,2) to (-4.5,2);
    \draw[red] (-2.6,2.3) node{$2^\ell$};
    \draw[thick,red,<->] (-4.5,2) to (-8.5,2);
    \draw[red] (-6.6,2.3) node{$2^\ell$};
  \end{scope}
  \begin{scope}[canvas is yz plane at x=5.5]
    \draw[ultra thick] (0,3) -- (-9,3);
    \draw[ultra thick,dashed] (0,-1) -- (-9,-1);
    \draw[thick,red,->] (-7.5,0) to (-7.5,3);
    \draw[thick,red,<-,dashed] (-7.5,-1) to (-7.5,0);
    \draw[red] (-7.8,-0.1) node{$2^\ell$};
  \end{scope}
      {
        \draw[fill=blue,opacity=0.02,canvas is yz plane at x=2] (0,-2) rectangle (-9,4);
        \draw[thick,blue,opacity=0.5,canvas is yz plane at x=2] (0,-2) rectangle (-9,4);
        \draw[fill=blue,opacity=0.02,canvas is xz plane at y=0] (1,-2) rectangle (2,4);
        \draw[thick,blue,opacity=0.5,canvas is xz plane at y=0] (1,-2) rectangle (2,4);
        \draw[fill=blue,opacity=0.02,canvas is xy plane at z=4] (1,0) rectangle (2,-9);
        \draw[thick,blue,opacity=0.5,canvas is xy plane at z=4] (1,0) rectangle (2,-9);
        \draw[very thick,opacity=0.5,blue,canvas is xy plane at z=0] (1,0) rectangle (2,-9);
      }
      {
        \draw[fill=black!20!green,opacity=0.02,canvas is yz plane at x=6] (0,-2) rectangle (-9,4);
        \draw[thick,black!20!green,opacity=0.5,canvas is yz plane at x=6] (0,-2) rectangle (-9,4);
        \draw[fill=black!20!green,opacity=0.02,canvas is xz plane at y=0] (5,-2) rectangle (6,4);
        \draw[thick,black!20!green,opacity=0.5,canvas is xz plane at y=0] (5,-2) rectangle (6,4);
        \draw[fill=black!20!green,opacity=0.02,canvas is xy plane at z=4] (5,0) rectangle (6,-9);
        \draw[thick,black!20!green,opacity=0.5,canvas is xy plane at z=4] (5,0) rectangle (6,-9);
        \draw[very thick,black!20!green,opacity=0.5,canvas is xy plane at z=0] (5,0) rectangle (6,-9);
      }
      \draw[opacity=0.3] (.6,2.5,0) node{$\ell = 1\ (\cdot,\uo)$};
      \draw[black!20!blue] (1.6,2.5,0) node{$\ell = 2\ (q_0,\uo)$};
      \draw[opacity=0.3] (2.6,2.5,0) node{$\ell = 1\ (\cdot,\uc)$}; 
      \draw[opacity=0.3] (3.6,2.5,0) node{$\ell = 3\ (\cdot,\uo)$}; 
      \draw[opacity=0.3] (4.6,2.5,0) node{$\ell = 1\ (\cdot,\uo)$}; 
      \draw[black!40!green] (5.6,2.5,0) node{$\ell = 2\ (q_0,\uc)$}; 
      \draw[opacity=0.3] (6.6,2.5,0) node{$\ell = 1\ (\cdot,\uc)$}; 
      \draw[opacity=0.3] (7.6,2.5,0) node{$\ell = 4\ (\cdot,\uo)$};
      \draw[canvas is xz plane at y=6.5, gray, ->] (5.5,0) -- ++(0,1) node[at end, above, gray] {$j$};
      \draw[canvas is xz plane at y=6.5, gray, ->] (5.5,0) -- ++(1,0) node[at end,below, gray] {$i$};
      \draw[canvas is yz plane at x=5.5, gray, ->] (6.5,0) -- ++(-1,0) node[at end,above, gray] {$k$};
\end{tikzpicture}

  \caption{A configuration of $Y_3$. To the left of each slice $(i,\cdot,\cdot)$ is its level $\ell$ and the values on Layers 1 and 2. We highlight two slices of level~2: at the front, marked by $(q_0, \uc$) with horizontal lines; at the back, marked by $(q_0, \uo)$ with vertical lines.}
\end{figure}

\begin{claim}\label{claim:y3-period-breaker}
A configuration of $Y_3$:
\begin{enumerate}
\item breaks every periodicity vector $(n,\cdot,\cdot)$ for $n \geq 1$.
\item every slice $(i,\cdot,\cdot)$ containing $(q_0,\uo)$ on the first two layers and corresponding to the lift of a single position of level $\ell \in [0,+\infty]$ in $T_\mathcal{M}$, is periodic with periods $(0,1,0)$ and $(0,0,2^\ell)$ but breaks every period $(\cdot,\cdot,n)$ for ${1 \leq n <2^\ell}$. The same is true with $(q_0,\uc)$ with vectors $(0,2^\ell,0)$ and $(0,0,1)$. 
\end{enumerate}
\end{claim}
\begin{claimproof}
\begin{enumerate}
\item the Toeplitzification of alternating $\uo$ and $\uc$ is aperiodic, so Layer 2 breaks all vectors $(n,\cdot,\cdot)$ for $n \geq 1$.

\item Layer 1 and 2 are lifted along the last two dimensions, so they cannot break any such vectors. Layer 3 is lifted along the second dimension so it is $(0,1,0)$-periodic, and breaks the required vectors from the last condition in the definition of $Y_2^{\uo}$. Layer 4 is $\square$ everywhere since it is not marked by $(q_0,\uc)$.\claimqedhere
\end{enumerate}\end{claimproof}

\subsection{Proof of the reduction \textbf{RS} \texorpdfstring{$\leq$}{<=} \textbf{AD}}\label{sec:3d-reduction}

\begin{lemma}\label{claim:y3-aperiodicity}
A configuration in $Y_3$ is aperiodic if, and only if, it corresponds to a run of $\mathcal{M}$ in which $q_0$ occurs infinitely often.
\end{lemma}

\begin{proof}
Using \Cref{claim:y3-period-breaker},
\begin{itemize}
\item Let $y \in Y_3$ be a configuration corresponding to a run of $\mathcal{M}$ that visit $q_0$ infinitely often.
  \begin{itemize}
  \item If $(q_0, \uo)$ appears at a level $\ell$, all vectors $(0,\cdot, n)$ for $1 \leq n<2^\ell$ are broken on Layer 3;
  \item Similarly for $(q_0, \uc)$ and vectors $(0, n, \cdot)$ on Layer 4.
    \end{itemize}
Therefore all vectors $(0,\cdot,\cdot)$ are broken at some level, and vectors $(n,\cdot,\cdot)$ are always broken for $n \geq 1$, so $y$ is an aperiodic configuration.
  
\item Let $y \in Y_3$ be a configuration corresponding to a run of $\mathcal{M}$ that does not visit $q_0$ after some time $N \in \N$. Let $a_\infty \in \dalphabet$ be the value on Layers 1 and 2 of the \emph{single} position of infinite level in $z$, if it exists.
  \begin{itemize}
  \item If $a_\infty \neq (q_0,\uc)$, positions marked by $(q_0,\uc)$ must be of level~$\leq N$, so $y$ is periodic of period $(0,2^{N},0)$.
  \item Similarly, if $a_\infty \neq (q_0,\uo)$, then $y$ is periodic of period $(0,0,2^{N})$.
  \end{itemize}
  
  All in all, $y$ is not aperiodic.\qedhere
\end{itemize}
\end{proof}

\begin{proof}[Case $d>3$]
We lift the previous construction and fill the additional dimensions with aperiodicity. More precisely, in the construction of $Y_3$, one of the dimension is always aperiodic, and the two others may or may not be periodic. Let $A$ be the $\Z^d$~lift of any $\Z^{d-2}$ aperiodic sofic subshift ($d-2 \geq 2$), and $Y_d$ the $\Z^d$ lift of $Y_3$. The cartesian product $A\times Y_d$ is aperiodic if and only if $Y_3$ is.
\end{proof}

\section{\texorpdfstring{$\Sigma_1^1$}{Sigma\_{}1\^{}1}-completeness for \texorpdfstring{$\Z^d$}{Z\^{}d} SFTs, \texorpdfstring{$d \geq 4$}{d >= 4}}
\label{sec:sft-highly-undecidable}

\begin{theorem}\label{th:4d-sft}
If $d \geq 4$, \textbf{AD} for $\Z^d$ SFTs is a $\Sigma_1^1$-complete problem.
\end{theorem}

As above, we prove $\Sigma_1^1$-hardness for $d=4$, and the result extends to $d>4$.

\subsection{Outline of the proof}

This proof has the same structure as \Cref{th:3d-sofic} with some adaptations for $\Z^4$~SFTs. We reduce to the problem \textbf{SR}: given $\mathcal{M}$ and $q_0$, we create an SFT $X_4$ that contains an aperiodic configuration if and only if $\mathcal{M}$ admits a run from the empty word which visits $q_0$ infinitely often. To do this, we use repeated lines along two dimensions (3 and 4) to break all periods up to a length controlled by a computation embedded in the configuration.

\begin{enumerate}
\item In \Cref{sec:2d-toeplitz}, we build $X_T^2$, a $\Z^2$ version of the Toeplitz structure $X_T$;
\item In \Cref{sec:3d-auxiliary}, we build auxiliary SFTs $X_3^{\uo}$ and $X_3^{\uc}$ (counterparts to $Y_2^{\uo}$ and $Y_2^{\uc}$);
\item In \Cref{sec:4d-final}, we build $X_4$ and prove the reduction.
\end{enumerate}

The main difference is that the finite type case requires an additional dimension to embed computations and some construction lines (dimensions 1 and 2). Remember that the subshift $T_\mathcal{M}$ of $\Z$ Toeplitzified sequences of states of runs of $\mathcal{M}$ is effective ; instead, we define an aperiodic $\Z^2$ version $T_\mathcal{M}^2$ that is sofic and aperiodic. Since $T_\mathcal{M}^2$ is the projection of an aperiodic $\Z^2$ SFT, we then add, as in the previous case,  two additional dimensions in which this SFT can be periodic or aperiodic, since the position of infinite level can break periods uncontrollably along at most one dimension.

Furthermore, to control the length of the vectors being broken, we need to measure distances between lines (as in $Y_3$). With SFTs, copying a distance from one dimension to another can only be done with diagonals. Therefore, instead of \emph{lines} to break periods, we use a more complex \emph{diagonal SFT} $D$ that we embed in slices $(i,\cdot,\cdot,\cdot)$ only on dimensions 2 and 3 (on symbols $\uo)$ or 2 and 4 (on symbols $\uc$). This way, the computation embedded in the first two dimensions can control the length of the broken periodicity vectors.

\subsection{\texorpdfstring{$T_\mathcal{M}^2$}{T\_M\^{}2}: \texorpdfstring{$\Z^2$}{Z\^{}2} Toeplitz corresponding to state sequences of \texorpdfstring{$\mathcal{M}$}{M}}\label{sec:2d-toeplitz}

\paragraph*{Binary Toeplitz structure}

In this section, we use a $\Z^2$ subshift $X_T^2$ on the alphabet $\{\ulc,\urc,\llc,\lrc,\sh,\sv\}$ whose structure is a two-dimensional analog of $X_T$. 

It is defined by the substitution $\sigma_2$ on the alphabet $\{ \gray{\ulc},\gray{\urc},\gray{\llc},\gray{\lrc},\gray{\sh},\gray{\sv}, \noop{\ulc},\noop{\urc},\noop{\llc},\noop{\lrc},\noop{\sh},\noop{\sv}\}$:
\[ \sigma_2 =\left\{\begin{array}{c}
      \gray{\sv}  \mapsto
      \begin{tikzpicture}[baseline={([yshift=-.5ex]current bounding box.center)}]
\draw[step=0.5, dotted, gray] (0,0) grid (1,1);
\node at (0.25, 0.75) {$\noop\sv$};
\node at (0.75, 0.75) {$\gray\sh$};
\node at (0.25, 0.25) {$\gray\sv$};
\node at (0.75, 0.25) {\Large$\gray\llc$};
\end{tikzpicture}
\hspace{1.3cm}
      \gray{\sh}  \mapsto
\begin{tikzpicture}[baseline={([yshift=-.5ex]current bounding box.center)}]
\draw[step=0.5, dotted, gray] (0,0) grid (1,1);
\node at (0.25, 0.75) {$\noop\sh$};
\node at (0.75, 0.75) {$\gray\sh$};
\node at (0.25, 0.25) {$\gray\sv$};
\node at (0.75, 0.25) {\Large$\gray\urc$};
\end{tikzpicture}
\hspace{1.3cm}

  \noop{l} \in \{ \noop{\sh},\noop{\sv} \}  \mapsto
  \begin{tikzpicture}[baseline={([yshift=-.5ex]current bounding box.center)}]
\draw[step=0.5, dotted, gray] (0,0) grid (1,1);
\node at (0.25, 0.75) {$\noop{l}$};
\node at (0.75, 0.75) {$\gray\sh$};
\node at (0.25, 0.25) {$\gray\sv$};
\node at (0.75, 0.25) {\Large$\gray\ulc$};
\end{tikzpicture}
   \\[.5cm]
  \gray{c} \in \{ \gray{\ulc},\gray{\urc},\gray{\lrc},\gray{\llc} \} \mapsto
\begin{tikzpicture}[baseline={([yshift=-.5ex]current bounding box.center)}]

\draw[step=0.5, dotted, gray] (0,0) grid (1,1);
\node at (0.25, 0.75) {$\noop{c}$};
\node at (0.75, 0.75) {$\gray\sh$};
\node at (0.25, 0.25) {$\gray\sv$};
\node at (0.75, 0.25) {\Large$\gray\lrc$};
\end{tikzpicture}
   \hspace{1.3cm}
 \noop{c} \in \{ \noop{\ulc},\noop{\urc},\noop{\llc},\noop{\lrc} \}  \mapsto
\begin{tikzpicture}[baseline={([yshift=-.5ex]current bounding box.center)}]
\draw[step=0.5, dotted, gray] (0,0) grid (1,1);
\node at (0.25, 0.75) {$\noop{c}$};
\node at (0.75, 0.75) {$\gray\sh$};
\node at (0.25, 0.25) {$\gray\sv$};
\node at (0.75, 0.25) {\Large$\gray\ulc$};
\end{tikzpicture}

  \end{array}\right.\]

As before, $X_{\sigma_2}$ is the subshift whose forbidden patterns are all the patterns which do not appear in the configuration $\sigma_2^\omega(\noop{\ulc})$ (any other seed symbol would do).

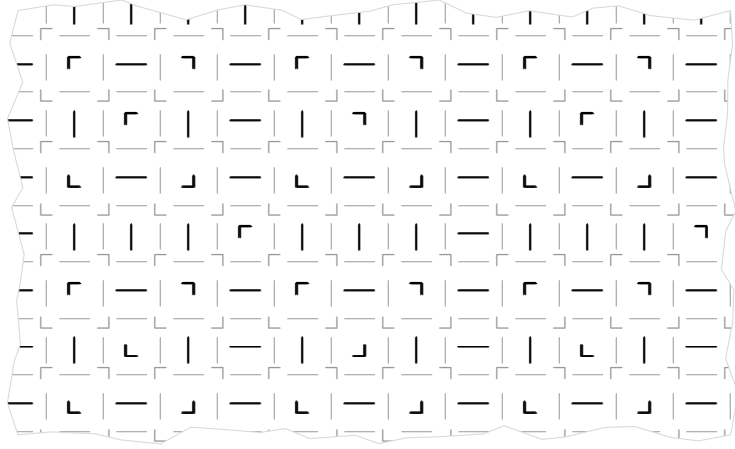
\begin{figure}[t]
  \centering
    \begin{tikzpicture}[scale=0.5, decoration={random steps,segment length=12pt,amplitude=4pt}, scale=0.75]
      \draw[decorate,black!15!white]  [path picture={
\draw (0.5,-0.5) node[black]{\Large $\ulc$};
\draw (1.5,-0.5) node[lipicsBulletGray]{$\sh$};
\draw (2.5,-0.5) node[black]{$\noop\sh$};
\draw (3.5,-0.5) node[lipicsBulletGray]{$\sh$};
\draw (4.5,-0.5) node[black]{$\noop\sh$};
\draw (5.5,-0.5) node[lipicsBulletGray]{$\sh$};
\draw (6.5,-0.5) node[black]{$\noop\sh$};
\draw (7.5,-0.5) node[lipicsBulletGray]{$\sh$};
\draw (8.5,-0.5) node[black]{$\noop\sh$};
\draw (9.5,-0.5) node[lipicsBulletGray]{$\sh$};
\draw (10.5,-0.5) node[black]{$\noop\sh$};
\draw (11.5,-0.5) node[lipicsBulletGray]{$\sh$};
\draw (12.5,-0.5) node[black]{$\noop\sh$};
\draw (13.5,-0.5) node[lipicsBulletGray]{$\sh$};
\draw (14.5,-0.5) node[black]{$\noop\sh$};
\draw (15.5,-0.5) node[lipicsBulletGray]{$\sh$};
\draw (16.5,-0.5) node[black]{$\noop\sh$};
\draw (17.5,-0.5) node[lipicsBulletGray]{$\sh$};
\draw (18.5,-0.5) node[black]{$\noop\sh$};
\draw (19.5,-0.5) node[lipicsBulletGray]{$\sh$};
\draw (20.5,-0.5) node[black]{$\noop\sh$};
\draw (21.5,-0.5) node[lipicsBulletGray]{$\sh$};
\draw (22.5,-0.5) node[black]{$\noop\sh$};
\draw (23.5,-0.5) node[lipicsBulletGray]{$\sh$};
\draw (24.5,-0.5) node[black]{$\noop\sh$};
\draw (25.5,-0.5) node[lipicsBulletGray]{$\sh$};
\draw (0.5,-1.5) node[lipicsBulletGray]{$\sv$};
\draw (1.5,-1.5) node[lipicsBulletGray]{\Large $\ulc$};
\draw (2.5,-1.5) node[lipicsBulletGray]{$\sv$};
\draw (3.5,-1.5) node[lipicsBulletGray]{\Large $\urc$};
\draw (4.5,-1.5) node[lipicsBulletGray]{$\sv$};
\draw (5.5,-1.5) node[lipicsBulletGray]{\Large $\ulc$};
\draw (6.5,-1.5) node[lipicsBulletGray]{$\sv$};
\draw (7.5,-1.5) node[lipicsBulletGray]{\Large $\urc$};
\draw (8.5,-1.5) node[lipicsBulletGray]{$\sv$};
\draw (9.5,-1.5) node[lipicsBulletGray]{\Large $\ulc$};
\draw (10.5,-1.5) node[lipicsBulletGray]{$\sv$};
\draw (11.5,-1.5) node[lipicsBulletGray]{\Large $\urc$};
\draw (12.5,-1.5) node[lipicsBulletGray]{$\sv$};
\draw (13.5,-1.5) node[lipicsBulletGray]{\Large $\ulc$};
\draw (14.5,-1.5) node[lipicsBulletGray]{$\sv$};
\draw (15.5,-1.5) node[lipicsBulletGray]{\Large $\urc$};
\draw (16.5,-1.5) node[lipicsBulletGray]{$\sv$};
\draw (17.5,-1.5) node[lipicsBulletGray]{\Large $\ulc$};
\draw (18.5,-1.5) node[lipicsBulletGray]{$\sv$};
\draw (19.5,-1.5) node[lipicsBulletGray]{\Large $\urc$};
\draw (20.5,-1.5) node[lipicsBulletGray]{$\sv$};
\draw (21.5,-1.5) node[lipicsBulletGray]{\Large $\ulc$};
\draw (22.5,-1.5) node[lipicsBulletGray]{$\sv$};
\draw (23.5,-1.5) node[lipicsBulletGray]{\Large $\urc$};
\draw (24.5,-1.5) node[lipicsBulletGray]{$\sv$};
\draw (25.5,-1.5) node[lipicsBulletGray]{\Large $\ulc$};
\draw (0.5,-2.5) node[black]{$\noop\sv$};
\draw (1.5,-2.5) node[lipicsBulletGray]{$\sh$};
\draw (2.5,-2.5) node[black]{\Large $\noop\ulc$};
\draw (3.5,-2.5) node[lipicsBulletGray]{$\sh$};
\draw (4.5,-2.5) node[black]{$\noop\sv$};
\draw (5.5,-2.5) node[lipicsBulletGray]{$\sh$};
\draw (6.5,-2.5) node[black]{\Large $\noop\urc$};
\draw (7.5,-2.5) node[lipicsBulletGray]{$\sh$};
\draw (8.5,-2.5) node[black]{$\noop\sv$};
\draw (9.5,-2.5) node[lipicsBulletGray]{$\sh$};
\draw (10.5,-2.5) node[black]{\Large $\noop\ulc$};
\draw (11.5,-2.5) node[lipicsBulletGray]{$\sh$};
\draw (12.5,-2.5) node[black]{$\noop\sv$};
\draw (13.5,-2.5) node[lipicsBulletGray]{$\sh$};
\draw (14.5,-2.5) node[black]{\Large $\noop\urc$};
\draw (15.5,-2.5) node[lipicsBulletGray]{$\sh$};
\draw (16.5,-2.5) node[black]{$\noop\sv$};
\draw (17.5,-2.5) node[lipicsBulletGray]{$\sh$};
\draw (18.5,-2.5) node[black]{\Large $\noop\ulc$};
\draw (19.5,-2.5) node[lipicsBulletGray]{$\sh$};
\draw (20.5,-2.5) node[black]{$\noop\sv$};
\draw (21.5,-2.5) node[lipicsBulletGray]{$\sh$};
\draw (22.5,-2.5) node[black]{\Large $\noop\urc$};
\draw (23.5,-2.5) node[lipicsBulletGray]{$\sh$};
\draw (24.5,-2.5) node[black]{$\noop\sv$};
\draw (25.5,-2.5) node[lipicsBulletGray]{$\sh$};
\draw (0.5,-3.5) node[lipicsBulletGray]{$\sv$};
\draw (1.5,-3.5) node[lipicsBulletGray]{\Large $\llc$};
\draw (2.5,-3.5) node[lipicsBulletGray]{$\sv$};
\draw (3.5,-3.5) node[lipicsBulletGray]{\Large $\lrc$};
\draw (4.5,-3.5) node[lipicsBulletGray]{$\sv$};
\draw (5.5,-3.5) node[lipicsBulletGray]{\Large $\llc$};
\draw (6.5,-3.5) node[lipicsBulletGray]{$\sv$};
\draw (7.5,-3.5) node[lipicsBulletGray]{\Large $\lrc$};
\draw (8.5,-3.5) node[lipicsBulletGray]{$\sv$};
\draw (9.5,-3.5) node[lipicsBulletGray]{\Large $\llc$};
\draw (10.5,-3.5) node[lipicsBulletGray]{$\sv$};
\draw (11.5,-3.5) node[lipicsBulletGray]{\Large $\lrc$};
\draw (12.5,-3.5) node[lipicsBulletGray]{$\sv$};
\draw (13.5,-3.5) node[lipicsBulletGray]{\Large $\llc$};
\draw (14.5,-3.5) node[lipicsBulletGray]{$\sv$};
\draw (15.5,-3.5) node[lipicsBulletGray]{\Large $\lrc$};
\draw (16.5,-3.5) node[lipicsBulletGray]{$\sv$};
\draw (17.5,-3.5) node[lipicsBulletGray]{\Large $\llc$};
\draw (18.5,-3.5) node[lipicsBulletGray]{$\sv$};
\draw (19.5,-3.5) node[lipicsBulletGray]{\Large $\lrc$};
\draw (20.5,-3.5) node[lipicsBulletGray]{$\sv$};
\draw (21.5,-3.5) node[lipicsBulletGray]{\Large $\llc$};
\draw (22.5,-3.5) node[lipicsBulletGray]{$\sv$};
\draw (23.5,-3.5) node[lipicsBulletGray]{\Large $\lrc$};
\draw (24.5,-3.5) node[lipicsBulletGray]{$\sv$};
\draw (25.5,-3.5) node[lipicsBulletGray]{\Large $\llc$};
\draw (0.5,-4.5) node[black]{$\noop\sv$};
\draw (1.5,-4.5) node[lipicsBulletGray]{$\sh$};
\draw (2.5,-4.5) node[black]{$\noop\sh$};
\draw (3.5,-4.5) node[lipicsBulletGray]{$\sh$};
\draw (4.5,-4.5) node[black]{\Large $\noop\ulc$};
\draw (5.5,-4.5) node[lipicsBulletGray]{$\sh$};
\draw (6.5,-4.5) node[black]{$\noop\sh$};
\draw (7.5,-4.5) node[lipicsBulletGray]{$\sh$};
\draw (8.5,-4.5) node[black]{$\noop\sv$};
\draw (9.5,-4.5) node[lipicsBulletGray]{$\sh$};
\draw (10.5,-4.5) node[black]{$\noop\sh$};
\draw (11.5,-4.5) node[lipicsBulletGray]{$\sh$};
\draw (12.5,-4.5) node[black]{\Large $\noop\urc$};
\draw (13.5,-4.5) node[lipicsBulletGray]{$\sh$};
\draw (14.5,-4.5) node[black]{$\noop\sh$};
\draw (15.5,-4.5) node[lipicsBulletGray]{$\sh$};
\draw (16.5,-4.5) node[black]{$\noop\sv$};
\draw (17.5,-4.5) node[lipicsBulletGray]{$\sh$};
\draw (18.5,-4.5) node[black]{$\noop\sh$};
\draw (19.5,-4.5) node[lipicsBulletGray]{$\sh$};
\draw (20.5,-4.5) node[black]{\Large $\noop\ulc$};
\draw (21.5,-4.5) node[lipicsBulletGray]{$\sh$};
\draw (22.5,-4.5) node[black]{$\noop\sh$};
\draw (23.5,-4.5) node[lipicsBulletGray]{$\sh$};
\draw (24.5,-4.5) node[black]{$\noop\sv$};
\draw (25.5,-4.5) node[lipicsBulletGray]{$\sh$};
\draw (0.5,-5.5) node[lipicsBulletGray]{$\sv$};
\draw (1.5,-5.5) node[lipicsBulletGray]{\Large $\ulc$};
\draw (2.5,-5.5) node[lipicsBulletGray]{$\sv$};
\draw (3.5,-5.5) node[lipicsBulletGray]{\Large $\urc$};
\draw (4.5,-5.5) node[lipicsBulletGray]{$\sv$};
\draw (5.5,-5.5) node[lipicsBulletGray]{\Large $\ulc$};
\draw (6.5,-5.5) node[lipicsBulletGray]{$\sv$};
\draw (7.5,-5.5) node[lipicsBulletGray]{\Large $\urc$};
\draw (8.5,-5.5) node[lipicsBulletGray]{$\sv$};
\draw (9.5,-5.5) node[lipicsBulletGray]{\Large $\ulc$};
\draw (10.5,-5.5) node[lipicsBulletGray]{$\sv$};
\draw (11.5,-5.5) node[lipicsBulletGray]{\Large $\urc$};
\draw (12.5,-5.5) node[lipicsBulletGray]{$\sv$};
\draw (13.5,-5.5) node[lipicsBulletGray]{\Large $\ulc$};
\draw (14.5,-5.5) node[lipicsBulletGray]{$\sv$};
\draw (15.5,-5.5) node[lipicsBulletGray]{\Large $\urc$};
\draw (16.5,-5.5) node[lipicsBulletGray]{$\sv$};
\draw (17.5,-5.5) node[lipicsBulletGray]{\Large $\ulc$};
\draw (18.5,-5.5) node[lipicsBulletGray]{$\sv$};
\draw (19.5,-5.5) node[lipicsBulletGray]{\Large $\urc$};
\draw (20.5,-5.5) node[lipicsBulletGray]{$\sv$};
\draw (21.5,-5.5) node[lipicsBulletGray]{\Large $\ulc$};
\draw (22.5,-5.5) node[lipicsBulletGray]{$\sv$};
\draw (23.5,-5.5) node[lipicsBulletGray]{\Large $\urc$};
\draw (24.5,-5.5) node[lipicsBulletGray]{$\sv$};
\draw (25.5,-5.5) node[lipicsBulletGray]{\Large $\ulc$};
\draw (0.5,-6.5) node[black]{$\noop\sv$};
\draw (1.5,-6.5) node[lipicsBulletGray]{$\sh$};
\draw (2.5,-6.5) node[black]{\Large $\noop\llc$};
\draw (3.5,-6.5) node[lipicsBulletGray]{$\sh$};
\draw (4.5,-6.5) node[black]{$\noop\sv$};
\draw (5.5,-6.5) node[lipicsBulletGray]{$\sh$};
\draw (6.5,-6.5) node[black]{\Large $\noop\lrc$};
\draw (7.5,-6.5) node[lipicsBulletGray]{$\sh$};
\draw (8.5,-6.5) node[black]{$\noop\sv$};
\draw (9.5,-6.5) node[lipicsBulletGray]{$\sh$};
\draw (10.5,-6.5) node[black]{\Large $\noop\llc$};
\draw (11.5,-6.5) node[lipicsBulletGray]{$\sh$};
\draw (12.5,-6.5) node[black]{$\noop\sv$};
\draw (13.5,-6.5) node[lipicsBulletGray]{$\sh$};
\draw (14.5,-6.5) node[black]{\Large $\noop\lrc$};
\draw (15.5,-6.5) node[lipicsBulletGray]{$\sh$};
\draw (16.5,-6.5) node[black]{$\noop\sv$};
\draw (17.5,-6.5) node[lipicsBulletGray]{$\sh$};
\draw (18.5,-6.5) node[black]{\Large $\noop\llc$};
\draw (19.5,-6.5) node[lipicsBulletGray]{$\sh$};
\draw (20.5,-6.5) node[black]{$\noop\sv$};
\draw (21.5,-6.5) node[lipicsBulletGray]{$\sh$};
\draw (22.5,-6.5) node[black]{\Large $\noop\lrc$};
\draw (23.5,-6.5) node[lipicsBulletGray]{$\sh$};
\draw (24.5,-6.5) node[black]{$\noop\sv$};
\draw (25.5,-6.5) node[lipicsBulletGray]{$\sh$};
\draw (0.5,-7.5) node[lipicsBulletGray]{$\sv$};
\draw (1.5,-7.5) node[lipicsBulletGray]{\Large $\llc$};
\draw (2.5,-7.5) node[lipicsBulletGray]{$\sv$};
\draw (3.5,-7.5) node[lipicsBulletGray]{\Large $\lrc$};
\draw (4.5,-7.5) node[lipicsBulletGray]{$\sv$};
\draw (5.5,-7.5) node[lipicsBulletGray]{\Large $\llc$};
\draw (6.5,-7.5) node[lipicsBulletGray]{$\sv$};
\draw (7.5,-7.5) node[lipicsBulletGray]{\Large $\lrc$};
\draw (8.5,-7.5) node[lipicsBulletGray]{$\sv$};
\draw (9.5,-7.5) node[lipicsBulletGray]{\Large $\llc$};
\draw (10.5,-7.5) node[lipicsBulletGray]{$\sv$};
\draw (11.5,-7.5) node[lipicsBulletGray]{\Large $\lrc$};
\draw (12.5,-7.5) node[lipicsBulletGray]{$\sv$};
\draw (13.5,-7.5) node[lipicsBulletGray]{\Large $\llc$};
\draw (14.5,-7.5) node[lipicsBulletGray]{$\sv$};
\draw (15.5,-7.5) node[lipicsBulletGray]{\Large $\lrc$};
\draw (16.5,-7.5) node[lipicsBulletGray]{$\sv$};
\draw (17.5,-7.5) node[lipicsBulletGray]{\Large $\llc$};
\draw (18.5,-7.5) node[lipicsBulletGray]{$\sv$};
\draw (19.5,-7.5) node[lipicsBulletGray]{\Large $\lrc$};
\draw (20.5,-7.5) node[lipicsBulletGray]{$\sv$};
\draw (21.5,-7.5) node[lipicsBulletGray]{\Large $\llc$};
\draw (22.5,-7.5) node[lipicsBulletGray]{$\sv$};
\draw (23.5,-7.5) node[lipicsBulletGray]{\Large $\lrc$};
\draw (24.5,-7.5) node[lipicsBulletGray]{$\sv$};
\draw (25.5,-7.5) node[lipicsBulletGray]{\Large $\llc$};
\draw (0.5,-8.5) node[black]{$\noop\sv$};
\draw (1.5,-8.5) node[lipicsBulletGray]{$\sh$};
\draw (2.5,-8.5) node[black]{$\noop\sh$};
\draw (3.5,-8.5) node[lipicsBulletGray]{$\sh$};
\draw (4.5,-8.5) node[black]{$\noop\sh$};
\draw (5.5,-8.5) node[lipicsBulletGray]{$\sh$};
\draw (6.5,-8.5) node[black]{$\noop\sh$};
\draw (7.5,-8.5) node[lipicsBulletGray]{$\sh$};
\draw (8.5,-8.5) node[black]{\Large $\noop\ulc$};
\draw (9.5,-8.5) node[lipicsBulletGray]{$\sh$};
\draw (10.5,-8.5) node[black]{$\noop\sh$};
\draw (11.5,-8.5) node[lipicsBulletGray]{$\sh$};
\draw (12.5,-8.5) node[black]{$\noop\sh$};
\draw (13.5,-8.5) node[lipicsBulletGray]{$\sh$};
\draw (14.5,-8.5) node[black]{$\noop\sh$};
\draw (15.5,-8.5) node[lipicsBulletGray]{$\sh$};
\draw (16.5,-8.5) node[black]{$\noop\sv$};
\draw (17.5,-8.5) node[lipicsBulletGray]{$\sh$};
\draw (18.5,-8.5) node[black]{$\noop\sh$};
\draw (19.5,-8.5) node[lipicsBulletGray]{$\sh$};
\draw (20.5,-8.5) node[black]{$\noop\sh$};
\draw (21.5,-8.5) node[lipicsBulletGray]{$\sh$};
\draw (22.5,-8.5) node[black]{$\noop\sh$};
\draw (23.5,-8.5) node[lipicsBulletGray]{$\sh$};
\draw (24.5,-8.5) node[black]{\Large $\noop\urc$};
\draw (25.5,-8.5) node[lipicsBulletGray]{$\sh$};
\draw (0.5,-9.5) node[lipicsBulletGray]{$\sv$};
\draw (1.5,-9.5) node[lipicsBulletGray]{\Large $\ulc$};
\draw (2.5,-9.5) node[lipicsBulletGray]{$\sv$};
\draw (3.5,-9.5) node[lipicsBulletGray]{\Large $\urc$};
\draw (4.5,-9.5) node[lipicsBulletGray]{$\sv$};
\draw (5.5,-9.5) node[lipicsBulletGray]{\Large $\ulc$};
\draw (6.5,-9.5) node[lipicsBulletGray]{$\sv$};
\draw (7.5,-9.5) node[lipicsBulletGray]{\Large $\urc$};
\draw (8.5,-9.5) node[lipicsBulletGray]{$\sv$};
\draw (9.5,-9.5) node[lipicsBulletGray]{\Large $\ulc$};
\draw (10.5,-9.5) node[lipicsBulletGray]{$\sv$};
\draw (11.5,-9.5) node[lipicsBulletGray]{\Large $\urc$};
\draw (12.5,-9.5) node[lipicsBulletGray]{$\sv$};
\draw (13.5,-9.5) node[lipicsBulletGray]{\Large $\ulc$};
\draw (14.5,-9.5) node[lipicsBulletGray]{$\sv$};
\draw (15.5,-9.5) node[lipicsBulletGray]{\Large $\urc$};
\draw (16.5,-9.5) node[lipicsBulletGray]{$\sv$};
\draw (17.5,-9.5) node[lipicsBulletGray]{\Large $\ulc$};
\draw (18.5,-9.5) node[lipicsBulletGray]{$\sv$};
\draw (19.5,-9.5) node[lipicsBulletGray]{\Large $\urc$};
\draw (20.5,-9.5) node[lipicsBulletGray]{$\sv$};
\draw (21.5,-9.5) node[lipicsBulletGray]{\Large $\ulc$};
\draw (22.5,-9.5) node[lipicsBulletGray]{$\sv$};
\draw (23.5,-9.5) node[lipicsBulletGray]{\Large $\urc$};
\draw (24.5,-9.5) node[lipicsBulletGray]{$\sv$};
\draw (25.5,-9.5) node[lipicsBulletGray]{\Large $\ulc$};
\draw (0.5,-10.5) node[black]{$\noop\sv$};
\draw (1.5,-10.5) node[lipicsBulletGray]{$\sh$};
\draw (2.5,-10.5) node[black]{\Large $\noop\ulc$};
\draw (3.5,-10.5) node[lipicsBulletGray]{$\sh$};
\draw (4.5,-10.5) node[black]{$\noop\sv$};
\draw (5.5,-10.5) node[lipicsBulletGray]{$\sh$};
\draw (6.5,-10.5) node[black]{\Large $\noop\urc$};
\draw (7.5,-10.5) node[lipicsBulletGray]{$\sh$};
\draw (8.5,-10.5) node[black]{$\noop\sv$};
\draw (9.5,-10.5) node[lipicsBulletGray]{$\sh$};
\draw (10.5,-10.5) node[black]{\Large $\noop\ulc$};
\draw (11.5,-10.5) node[lipicsBulletGray]{$\sh$};
\draw (12.5,-10.5) node[black]{$\noop\sv$};
\draw (13.5,-10.5) node[lipicsBulletGray]{$\sh$};
\draw (14.5,-10.5) node[black]{\Large $\noop\urc$};
\draw (15.5,-10.5) node[lipicsBulletGray]{$\sh$};
\draw (16.5,-10.5) node[black]{$\noop\sv$};
\draw (17.5,-10.5) node[lipicsBulletGray]{$\sh$};
\draw (18.5,-10.5) node[black]{\Large $\noop\ulc$};
\draw (19.5,-10.5) node[lipicsBulletGray]{$\sh$};
\draw (20.5,-10.5) node[black]{$\noop\sv$};
\draw (21.5,-10.5) node[lipicsBulletGray]{$\sh$};
\draw (22.5,-10.5) node[black]{\Large $\noop\urc$};
\draw (23.5,-10.5) node[lipicsBulletGray]{$\sh$};
\draw (24.5,-10.5) node[black]{$\noop\sv$};
\draw (25.5,-10.5) node[lipicsBulletGray]{$\sh$};
\draw (0.5,-11.5) node[lipicsBulletGray]{$\sv$};
\draw (1.5,-11.5) node[lipicsBulletGray]{\Large $\llc$};
\draw (2.5,-11.5) node[lipicsBulletGray]{$\sv$};
\draw (3.5,-11.5) node[lipicsBulletGray]{\Large $\lrc$};
\draw (4.5,-11.5) node[lipicsBulletGray]{$\sv$};
\draw (5.5,-11.5) node[lipicsBulletGray]{\Large $\llc$};
\draw (6.5,-11.5) node[lipicsBulletGray]{$\sv$};
\draw (7.5,-11.5) node[lipicsBulletGray]{\Large $\lrc$};
\draw (8.5,-11.5) node[lipicsBulletGray]{$\sv$};
\draw (9.5,-11.5) node[lipicsBulletGray]{\Large $\llc$};
\draw (10.5,-11.5) node[lipicsBulletGray]{$\sv$};
\draw (11.5,-11.5) node[lipicsBulletGray]{\Large $\lrc$};
\draw (12.5,-11.5) node[lipicsBulletGray]{$\sv$};
\draw (13.5,-11.5) node[lipicsBulletGray]{\Large $\llc$};
\draw (14.5,-11.5) node[lipicsBulletGray]{$\sv$};
\draw (15.5,-11.5) node[lipicsBulletGray]{\Large $\lrc$};
\draw (16.5,-11.5) node[lipicsBulletGray]{$\sv$};
\draw (17.5,-11.5) node[lipicsBulletGray]{\Large $\llc$};
\draw (18.5,-11.5) node[lipicsBulletGray]{$\sv$};
\draw (19.5,-11.5) node[lipicsBulletGray]{\Large $\lrc$};
\draw (20.5,-11.5) node[lipicsBulletGray]{$\sv$};
\draw (21.5,-11.5) node[lipicsBulletGray]{\Large $\llc$};
\draw (22.5,-11.5) node[lipicsBulletGray]{$\sv$};
\draw (23.5,-11.5) node[lipicsBulletGray]{\Large $\lrc$};
\draw (24.5,-11.5) node[lipicsBulletGray]{$\sv$};
\draw (25.5,-11.5) node[lipicsBulletGray]{\Large $\llc$};
\draw (0.5,-12.5) node[black]{$\noop\sv$};
\draw (1.5,-12.5) node[lipicsBulletGray]{$\sh$};
\draw (2.5,-12.5) node[black]{$\noop\sh$};
\draw (3.5,-12.5) node[lipicsBulletGray]{$\sh$};
\draw (4.5,-12.5) node[black]{\Large $\noop\llc$};
\draw (5.5,-12.5) node[lipicsBulletGray]{$\sh$};
\draw (6.5,-12.5) node[black]{$\noop\sh$};
\draw (7.5,-12.5) node[lipicsBulletGray]{$\sh$};
\draw (8.5,-12.5) node[black]{$\noop\sv$};
\draw (9.5,-12.5) node[lipicsBulletGray]{$\sh$};
\draw (10.5,-12.5) node[black]{$\noop\sh$};
\draw (11.5,-12.5) node[lipicsBulletGray]{$\sh$};
\draw (12.5,-12.5) node[black]{\Large $\noop\lrc$};
\draw (13.5,-12.5) node[lipicsBulletGray]{$\sh$};
\draw (14.5,-12.5) node[black]{$\noop\sh$};
\draw (15.5,-12.5) node[lipicsBulletGray]{$\sh$};
\draw (16.5,-12.5) node[black]{$\noop\sv$};
\draw (17.5,-12.5) node[lipicsBulletGray]{$\sh$};
\draw (18.5,-12.5) node[black]{$\noop\sh$};
\draw (19.5,-12.5) node[lipicsBulletGray]{$\sh$};
\draw (20.5,-12.5) node[black]{\Large $\noop\llc$};
\draw (21.5,-12.5) node[lipicsBulletGray]{$\sh$};
\draw (22.5,-12.5) node[black]{$\noop\sh$};
\draw (23.5,-12.5) node[lipicsBulletGray]{$\sh$};
\draw (24.5,-12.5) node[black]{$\noop\sv$};
\draw (25.5,-12.5) node[lipicsBulletGray]{$\sh$};
\draw (0.5,-13.5) node[lipicsBulletGray]{$\sv$};
\draw (1.5,-13.5) node[lipicsBulletGray]{\Large $\ulc$};
\draw (2.5,-13.5) node[lipicsBulletGray]{$\sv$};
\draw (3.5,-13.5) node[lipicsBulletGray]{\Large $\urc$};
\draw (4.5,-13.5) node[lipicsBulletGray]{$\sv$};
\draw (5.5,-13.5) node[lipicsBulletGray]{\Large $\ulc$};
\draw (6.5,-13.5) node[lipicsBulletGray]{$\sv$};
\draw (7.5,-13.5) node[lipicsBulletGray]{\Large $\urc$};
\draw (8.5,-13.5) node[lipicsBulletGray]{$\sv$};
\draw (9.5,-13.5) node[lipicsBulletGray]{\Large $\ulc$};
\draw (10.5,-13.5) node[lipicsBulletGray]{$\sv$};
\draw (11.5,-13.5) node[lipicsBulletGray]{\Large $\urc$};
\draw (12.5,-13.5) node[lipicsBulletGray]{$\sv$};
\draw (13.5,-13.5) node[lipicsBulletGray]{\Large $\ulc$};
\draw (14.5,-13.5) node[lipicsBulletGray]{$\sv$};
\draw (15.5,-13.5) node[lipicsBulletGray]{\Large $\urc$};
\draw (16.5,-13.5) node[lipicsBulletGray]{$\sv$};
\draw (17.5,-13.5) node[lipicsBulletGray]{\Large $\ulc$};
\draw (18.5,-13.5) node[lipicsBulletGray]{$\sv$};
\draw (19.5,-13.5) node[lipicsBulletGray]{\Large $\urc$};
\draw (20.5,-13.5) node[lipicsBulletGray]{$\sv$};
\draw (21.5,-13.5) node[lipicsBulletGray]{\Large $\ulc$};
\draw (22.5,-13.5) node[lipicsBulletGray]{$\sv$};
\draw (23.5,-13.5) node[lipicsBulletGray]{\Large $\urc$};
\draw (24.5,-13.5) node[lipicsBulletGray]{$\sv$};
\draw (25.5,-13.5) node[lipicsBulletGray]{\Large $\ulc$};
\draw (0.5,-14.5) node[black]{$\noop\sv$};
\draw (1.5,-14.5) node[lipicsBulletGray]{$\sh$};
\draw (2.5,-14.5) node[black]{\Large $\noop\llc$};
\draw (3.5,-14.5) node[lipicsBulletGray]{$\sh$};
\draw (4.5,-14.5) node[black]{$\noop\sv$};
\draw (5.5,-14.5) node[lipicsBulletGray]{$\sh$};
\draw (6.5,-14.5) node[black]{\Large $\noop\lrc$};
\draw (7.5,-14.5) node[lipicsBulletGray]{$\sh$};
\draw (8.5,-14.5) node[black]{$\noop\sv$};
\draw (9.5,-14.5) node[lipicsBulletGray]{$\sh$};
\draw (10.5,-14.5) node[black]{\Large $\noop\llc$};
\draw (11.5,-14.5) node[lipicsBulletGray]{$\sh$};
\draw (12.5,-14.5) node[black]{$\noop\sv$};
\draw (13.5,-14.5) node[lipicsBulletGray]{$\sh$};
\draw (14.5,-14.5) node[black]{\Large $\noop\lrc$};
\draw (15.5,-14.5) node[lipicsBulletGray]{$\sh$};
\draw (16.5,-14.5) node[black]{$\noop\sv$};
\draw (17.5,-14.5) node[lipicsBulletGray]{$\sh$};
\draw (18.5,-14.5) node[black]{\Large $\noop\llc$};
\draw (19.5,-14.5) node[lipicsBulletGray]{$\sh$};
\draw (20.5,-14.5) node[black]{$\noop\sv$};
\draw (21.5,-14.5) node[lipicsBulletGray]{$\sh$};
\draw (22.5,-14.5) node[black]{\Large $\noop\lrc$};
\draw (23.5,-14.5) node[lipicsBulletGray]{$\sh$};
\draw (24.5,-14.5) node[black]{$\noop\sv$};
\draw (25.5,-14.5) node[lipicsBulletGray]{$\sh$};
\draw (0.5,-15.5) node[lipicsBulletGray]{$\sv$};
\draw (1.5,-15.5) node[lipicsBulletGray]{\Large $\llc$};
\draw (2.5,-15.5) node[lipicsBulletGray]{$\sv$};
\draw (3.5,-15.5) node[lipicsBulletGray]{\Large $\lrc$};
\draw (4.5,-15.5) node[lipicsBulletGray]{$\sv$};
\draw (5.5,-15.5) node[lipicsBulletGray]{\Large $\llc$};
\draw (6.5,-15.5) node[lipicsBulletGray]{$\sv$};
\draw (7.5,-15.5) node[lipicsBulletGray]{\Large $\lrc$};
\draw (8.5,-15.5) node[lipicsBulletGray]{$\sv$};
\draw (9.5,-15.5) node[lipicsBulletGray]{\Large $\llc$};
\draw (10.5,-15.5) node[lipicsBulletGray]{$\sv$};
\draw (11.5,-15.5) node[lipicsBulletGray]{\Large $\lrc$};
\draw (12.5,-15.5) node[lipicsBulletGray]{$\sv$};
\draw (13.5,-15.5) node[lipicsBulletGray]{\Large $\llc$};
\draw (14.5,-15.5) node[lipicsBulletGray]{$\sv$};
\draw (15.5,-15.5) node[lipicsBulletGray]{\Large $\lrc$};
\draw (16.5,-15.5) node[lipicsBulletGray]{$\sv$};
\draw (17.5,-15.5) node[lipicsBulletGray]{\Large $\llc$};
\draw (18.5,-15.5) node[lipicsBulletGray]{$\sv$};
\draw (19.5,-15.5) node[lipicsBulletGray]{\Large $\lrc$};
\draw (20.5,-15.5) node[lipicsBulletGray]{$\sv$};
\draw (21.5,-15.5) node[lipicsBulletGray]{\Large $\llc$};
\draw (22.5,-15.5) node[lipicsBulletGray]{$\sv$};
\draw (23.5,-15.5) node[lipicsBulletGray]{\Large $\lrc$};
\draw (24.5,-15.5) node[lipicsBulletGray]{$\sv$};
\draw (25.5,-15.5) node[lipicsBulletGray]{\Large $\llc$};}] 
(0.5,-15.5) rectangle (25.5,-0.5);
  \end{tikzpicture}
  \caption{A configuration of $X_{\sigma_2}$.}
\end{figure}

\begin{definition}[binary bi-Toeplitz structure]\label{def:bi-toeplitz}
  $X_T^2$ is the color-forgetting projection of $X_{\sigma_2}$ on the alphabet $\{\ulc ,\urc ,\llc ,\lrc ,\sh ,\sv \}$.
\end{definition}

As the substitution $\sigma_2$ is deterministic, $X_T^2$ is a sofic subshift by \cite[Theorem 4.1]{1989-Mozes}.

In a configuration of $X_T^2$, ignoring symbols $\sh$ and $\sv$:\hypertarget{points-2d-toeplitz}{}
\begin{enumerate}
\item corner symbols $\{\ulc,\urc,\llc,\lrc\}$ can be grouped together to form squares. A square is of level~$\ell$ if its edges have length $2^\ell$. There may exist a single corner in an otherwise blank line or column: we say it is part of a square of infinite level;
\item each line only contains symbols in either $\{\ulc, \urc\}$ or $\{\llc,\lrc\}$, all of the same level. If a line does not contain any corner, its level is said to be infinite;
\item the vertical distance between two consecutive lines of the same level $\ell$ is $2^\ell$, and those lines contain the same symbols.
\end{enumerate}
The corresponding statements hold for columns.

\paragraph*{\texorpdfstring{$T_\mathcal{M}^2$}{T\_M\^{}2}: \texorpdfstring{$\Z^2$}{Z\^{}2} Toeplitzification of sequences of states of $\mathcal{M}$}

\begin{definition}[$\Z^2$ Toeplitzification of a set of sequences]\label{bi-toeplitzification}
Given a set of sequences $A \subseteq \Sigma^\N$, we define the corresponding $\Z^2$ Toeplitzified subshift $T_A^2$ on the alphabet $\Sigma_T = \Sigma \times \{\uo,\uc\} \times \{\ulc,\llc,\urc,\lrc,\sh,\sv\}$ as:
\[ T_A^2 = \left\{ x \in (\Sigma_T)^{\Z^2}:
  \begin{array}{l}
    \pi_{1,2}(x) \in (T_A)^\uparrow, \pi_3(x) \in X_T^2 \\
    \forall i,j \in \Z, \begin{array}{l}
      \pi_3(x_{i,j}) \in \{\ulc,\llc\} \implies \pi_2(x_{i,j}) = \uo \\
      \pi_3(x_{i,j}) \in \{\urc,\lrc\} \implies \pi_2(x_{i,j}) = \uc
      \end{array} 
  \end{array}\right\} \]
\end{definition}

In other words, $T_A^2$ superimposes the structure of a $\Z$ Toeplitzification with the $\Z^2$ structure $X_T^2$ we define above. As before, we denote $T_\mathcal{M}^2 := T_{S_\mathcal{M}}^2$ where $S_\mathcal{M}$ of the set of sequences of states corresponding to runs of $\mathcal M$.

\begin{lemma}\label{bitoeplitz-sofic}
  $T_\mathcal{M}^2$ is a non-empty sofic subshift.
\end{lemma}
\begin{proof}$T_\mathcal{M}^2$ is non-empty as arrows in $T_\mathcal{M}$ can be aligned with corners in $X_T^2$: arrows of level $\ell$ as well as columns of level $\ell$ have period $2^\ell$. For~soficness:
  \begin{itemize}
  \item Layers 1 and 2 are $\Z^2$ lifts of configurations of $T_\mathcal{M}$, which is an effective subshift (\Cref{lemma:tm-effective}). By \Cref{thold:effective-sofic}, Layers 1 and 2 form a sofic subshift;
  \item Layer 3 is composed of configurations of $X_T^2$, which is a sofic subshift;
  \item the additional condition defining $T_\mathcal{M}^2$ (synchronizing Layers 2 and 3) is of finite type.\qedhere
  \end{itemize}\end{proof}

\subsection{Auxiliary \texorpdfstring{$\Z^2$}{Z\^{}2} and \texorpdfstring{$\Z^3$}{Z\^{}3} SFTs}\label{sec:3d-auxiliary}

\paragraph*{The diagonal SFT $D$}

The diagonal SFT $D$ is defined by adjacent matching patterns on the alphabet $\Sigma_D$:

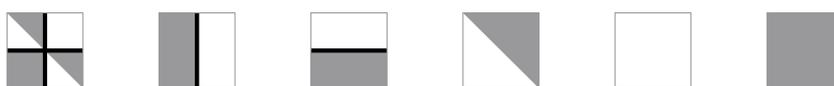
\begin{center}
\begin{tikzpicture}
  \fill[lipicsBulletGray] (0,0) rectangle (1,1);
  \fill[white] (0,1) -- (0.5,0.5) -- (0,0.5) -- cycle;
  \fill[white] (0.5,0.5) rectangle (1,1);
  \fill[white] (0.5,0.5) -- (1,0) -- (0.5,0) -- cycle;
  \draw[ultra thick] (0,0.5) -- (1,0.5);
  \draw[ultra thick] (0.5,0) -- (0.5,1);
  \draw[lipicsBulletGray] (0,0) rectangle (1,1);

  \fill[lipicsBulletGray] (2,0) rectangle (2.5,1);
  \fill[white] (2.5,0) rectangle (3,1);
  \draw[ultra thick] (2.5,0) -- (2.5,1);
   \draw[lipicsBulletGray] (2,0) rectangle (3,1);

  \fill[lipicsBulletGray] (4,0) rectangle (5,0.5);
  \fill[white] (4,0.5) rectangle (5,1);
  \draw[ultra thick] (4,0.5) -- (5,0.5);
  \draw[lipicsBulletGray] (4,0) rectangle (5,1);

  \fill[white] (6,1) -- (7,0) -- (6,0) -- cycle;
  \fill[lipicsBulletGray] (6,1) -- (7,1) -- (7,0) -- cycle;
  \draw[lipicsBulletGray] (6,0) rectangle (7,1);

  \fill[white] (8,0) rectangle (9,1);
  \draw[lipicsBulletGray] (8,0) rectangle (9,1);

  \filldraw[lipicsBulletGray] (10,0) rectangle (11,1);
\end{tikzpicture}
\end{center}

A configuration of $D$ containing two parallel black lines is in fact periodic and consists of repeated squares. $D$ also contains configurations with $0$ or $1$ black line.

\begin{figure}[t]
  \centering
  \begin{tikzpicture}[scale = 0.5, decoration={random steps,segment length=12pt,amplitude=4pt}]
\draw[decorate,black!15!white]  [path picture={
\fill[white] (-2.5,-2.5) -- (2.5,-2.5) -- (-2.5,2.5) -- cycle;
\fill[lipicsBulletGray] (-2.5,2.5) -- (2.5,-2.5) -- (2.5,2.5) -- cycle;
\fill[white] (-2.5,2.5) -- (2.5,2.5) -- (-2.5,7.5) -- cycle;
\fill[lipicsBulletGray] (-2.5,7.5) -- (2.5,2.5) -- (2.5,7.5) -- cycle;
\fill[white] (-2.5,7.5) -- (2.5,7.5) -- (-2.5,12.5) -- cycle;
\fill[lipicsBulletGray] (-2.5,12.5) -- (2.5,7.5) -- (2.5,12.5) -- cycle;
\fill[white] (2.5,-2.5) -- (7.5,-2.5) -- (2.5,2.5) -- cycle;
\fill[lipicsBulletGray] (2.5,2.5) -- (7.5,-2.5) -- (7.5,2.5) -- cycle;
\fill[white] (2.5,2.5) -- (7.5,2.5) -- (2.5,7.5) -- cycle;
\fill[lipicsBulletGray] (2.5,7.5) -- (7.5,2.5) -- (7.5,7.5) -- cycle;
\fill[white] (2.5,7.5) -- (7.5,7.5) -- (2.5,12.5) -- cycle;
\fill[lipicsBulletGray] (2.5,12.5) -- (7.5,7.5) -- (7.5,12.5) -- cycle;
\fill[white] (7.5,-2.5) -- (12.5,-2.5) -- (7.5,2.5) -- cycle;
\fill[lipicsBulletGray] (7.5,2.5) -- (12.5,-2.5) -- (12.5,2.5) -- cycle;
\fill[white] (7.5,2.5) -- (12.5,2.5) -- (7.5,7.5) -- cycle;
\fill[lipicsBulletGray] (7.5,7.5) -- (12.5,2.5) -- (12.5,7.5) -- cycle;
\fill[white] (7.5,7.5) -- (12.5,7.5) -- (7.5,12.5) -- cycle;
\fill[lipicsBulletGray] (7.5,12.5) -- (12.5,7.5) -- (12.5,12.5) -- cycle;
\fill[white] (12.5,-2.5) -- (17.5,-2.5) -- (12.5,2.5) -- cycle;
\fill[lipicsBulletGray] (12.5,2.5) -- (17.5,-2.5) -- (17.5,2.5) -- cycle;
\fill[white] (12.5,2.5) -- (17.5,2.5) -- (12.5,7.5) -- cycle;
\fill[lipicsBulletGray] (12.5,7.5) -- (17.5,2.5) -- (17.5,7.5) -- cycle;
\fill[white] (12.5,7.5) -- (17.5,7.5) -- (12.5,12.5) -- cycle;
\fill[lipicsBulletGray] (12.5,12.5) -- (17.5,7.5) -- (17.5,12.5) -- cycle;
\draw[black!15!white] (0,0) grid (17,10);
\draw[black,xshift=-2.5cm,yshift=-2.5cm,step=5,ultra thick] (-0.5,-0.5) grid (19.5,14.5);
}]
(0.5,0.5) -- (15.5,0.5) -- (15.5,9.5) -- (0.5,9.5) -- cycle;
\end{tikzpicture}

  \caption{A configuration of $D$.}
\end{figure}
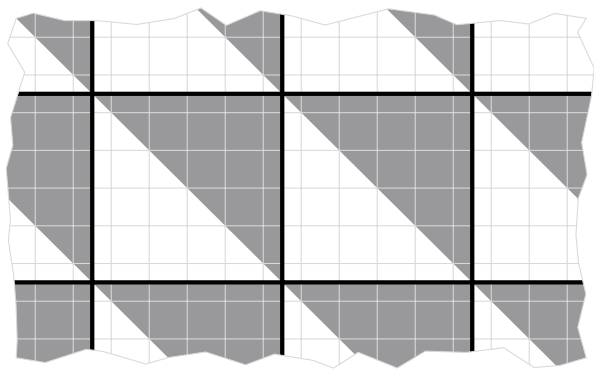

\paragraph*{\texorpdfstring{$\Z^3$}{Z\^{}3} SFTs $X_3^\uo$ and $X_3^\uc$}

This section is written for $X_3^\uo$; it applies to $X_3^\uc$ by flipping the arrows and corners. As $T_\mathcal{M}^2$ is a 3-layer sofic subshift, it is the projection of some 4-layer $\Z^2$ SFT $X_2$. To create $X_3^\uo$, we lift $X_2$ then add an additional layer for $D$:
\begin{itemize}
\item Layers 1 to 4: $\Z^3$ lift of $X_2$. In other words, $\pi_{1,2,3}(X_2)$ is a $\Z^3$ lift of $T_\mathcal{M}^2$;
\item Layer 5 : each slice $(i,\cdot,\cdot)$ contains a configuration $d^i \in D$. If the slice has $(q_0,\uo)$ on its first two layers, then the configuration $d^i$ is ``synchronized'' with the underlying configuration of $T_\mathcal{M}^2$ on Layer 3, that is: vertical lines
  \begin{tikzpicture}[baseline=0.3em]
    \draw[very thick] (0.2,0) -- (0.2,0.4);
  \end{tikzpicture}
  in $d^i$ only appear on Layer 5 at positions marked by corners $\ulc,\llc$ on Layer 3. Otherwise, the slice on Layer 5 is left blank.
\end{itemize}

\begin{figure}[t]
  \centering
  \begin{tikzpicture}[scale = 0.8, x={(-45:1cm)},y={(195:1cm)},z={(90:1cm)}]
  \begin{scope}[canvas is xy plane at z=0]
    \clip (-0.2,0.2) rectangle (5.2,-9.2);
    \draw[gray] (-1,1) grid (6,-10);
  \end{scope}
  \begin{scope}[canvas is yz plane at x=1.5,yshift=-1cm]
    \clip (0.2,1) rectangle (-9.2,6);
    \begin{scope}
      \fill[white] (1.5,-2.5) -- (-2.5,-2.5) -- (1.5,1.5) -- cycle;
      \fill[lipicsBulletGray] (1.5,1.5) -- (-2.5,-2.5) -- (-2.5,1.5) -- cycle;
      \fill[white] (1.5,1.5) -- (-2.5,1.5) -- (1.5,5.5) -- cycle;
      \fill[lipicsBulletGray] (1.5,5.5) -- (-2.5,1.5) -- (-2.5,5.5) -- cycle;
      \fill[white] (1.5,5.5) -- (-2.5,5.5) -- (1.5,9.5) -- cycle;
      \fill[lipicsBulletGray] (1.5,9.5) -- (-2.5,5.5) -- (-2.5,9.5) -- cycle;
      \fill[white] (-2.5,-2.5) -- (-6.5,-2.5) -- (-2.5,1.5) -- cycle;
      \fill[lipicsBulletGray] (-2.5,1.5) -- (-6.5,-2.5) -- (-6.5,1.5) -- cycle;
      \fill[white] (-2.5,1.5) -- (-6.5,1.5) -- (-2.5,5.5) -- cycle;
      \fill[lipicsBulletGray] (-2.5,5.5) -- (-6.5,1.5) -- (-6.5,5.5) -- cycle;
      \fill[white] (-2.5,5.5) -- (-6.5,5.5) -- (-2.5,9.5) -- cycle;
      \fill[lipicsBulletGray] (-2.5,9.5) -- (-6.5,5.5) -- (-6.5,9.5) -- cycle;
      \fill[white] (-6.5,-2.5) -- (-10.5,-2.5) -- (-6.5,1.5) -- cycle;
      \fill[lipicsBulletGray] (-6.5,1.5) -- (-10.5,-2.5) -- (-10.5,1.5) -- cycle;
      \fill[white] (-6.5,1.5) -- (-10.5,1.5) -- (-6.5,5.5) -- cycle;
      \fill[lipicsBulletGray] (-6.5,5.5) -- (-10.5,1.5) -- (-10.5,5.5) -- cycle;
      \fill[white] (-6.5,5.5) -- (-10.5,5.5) -- (-6.5,9.5) -- cycle;
      \draw[black!15!white] (1,-3) grid (-10,6);
      \draw[black,xshift=-2.5cm,yshift=-2.5cm,step=4,ultra thick] (5,-3) grid (-10,12);
    \end{scope}
  \end{scope}
  \begin{scope}[opacity=0.3,canvas is yz plane at x=1.5,yshift=-1cm]
    \clip (0.2,-1) rectangle (-9.2,1);
      \fill[white] (1.5,-2.5) -- (-2.5,-2.5) -- (1.5,1.5) -- cycle;
      \fill[lipicsBulletGray] (1.5,1.5) -- (-2.5,-2.5) -- (-2.5,1.5) -- cycle;
      \fill[white] (1.5,1.5) -- (-2.5,1.5) -- (1.5,5.5) -- cycle;
      \fill[lipicsBulletGray] (1.5,5.5) -- (-2.5,1.5) -- (-2.5,5.5) -- cycle;
      \fill[white] (1.5,5.5) -- (-2.5,5.5) -- (1.5,9.5) -- cycle;
      \fill[lipicsBulletGray] (1.5,9.5) -- (-2.5,5.5) -- (-2.5,9.5) -- cycle;
      \fill[white] (-2.5,-2.5) -- (-6.5,-2.5) -- (-2.5,1.5) -- cycle;
      \fill[lipicsBulletGray] (-2.5,1.5) -- (-6.5,-2.5) -- (-6.5,1.5) -- cycle;
      \fill[white] (-2.5,1.5) -- (-6.5,1.5) -- (-2.5,5.5) -- cycle;
      \fill[lipicsBulletGray] (-2.5,5.5) -- (-6.5,1.5) -- (-6.5,5.5) -- cycle;
      \fill[white] (-2.5,5.5) -- (-6.5,5.5) -- (-2.5,9.5) -- cycle;
      \fill[lipicsBulletGray] (-2.5,9.5) -- (-6.5,5.5) -- (-6.5,9.5) -- cycle;
      \fill[white] (-6.5,-2.5) -- (-10.5,-2.5) -- (-6.5,1.5) -- cycle;
      \fill[lipicsBulletGray] (-6.5,1.5) -- (-10.5,-2.5) -- (-10.5,1.5) -- cycle;
      \fill[white] (-6.5,1.5) -- (-10.5,1.5) -- (-6.5,5.5) -- cycle;
      \fill[lipicsBulletGray] (-6.5,5.5) -- (-10.5,1.5) -- (-10.5,5.5) -- cycle;
      \fill[white] (-6.5,5.5) -- (-10.5,5.5) -- (-6.5,9.5) -- cycle;
      \draw[black!15!white] (1,-3) grid (-10,6);
      \draw[black,xshift=-2.5cm,yshift=-2.5cm,step=4,ultra thick] (5,-3) grid (-10,12);
  \end{scope}
  \begin{scope}[blue,opacity=0.2,canvas is xy plane at z=0, transform shape,xshift=-0.85cm]
    \clip (0.85,0.2) rectangle (2.2,-9.2);
    \draw (0.5,-0.5) node[]{\Huge $\ulc$};
    \draw (1.5,-1.5) node[]{\Huge $\ulc$};
    \draw (1.5,-3.5) node[]{\Huge $\llc$};
    \draw (1.5,-5.5) node[]{\Huge $\ulc$};
    \draw (1.5,-7.5) node[]{\Huge $\llc$};
  \end{scope}
  \begin{scope}[blue,canvas is xy plane at z=0, transform shape,xshift=-0.85cm]
    \clip (2.2,0.2) rectangle (5.85,-9.2);
    \draw (3.5,-1.5) node[]{\Huge $\urc$};
    \draw (5.5,-1.5) node[]{\Huge $\ulc$};
    \draw (7.5,-1.5) node[]{\Huge $\urc$};
    \draw (2.5,-2.5) node[]{\Huge $\ulc$};
    \draw (6.5,-2.5) node[]{\Huge $\urc$};
    \draw (3.5,-3.5) node[]{\Huge $\lrc$};
    \draw (5.5,-3.5) node[]{\Huge $\llc$};
    \draw (7.5,-3.5) node[]{\Huge $\lrc$};
    \draw (4.5,-4.5) node[]{\Huge $\ulc$};
    \draw (3.5,-5.5) node[]{\Huge $\urc$};
    \draw (5.5,-5.5) node[]{\Huge $\ulc$};
    \draw (7.5,-5.5) node[]{\Huge $\urc$};
    \draw (2.5,-6.5) node[]{\Huge $\llc$};
    \draw (6.5,-6.5) node[]{\Huge $\lrc$};
    \draw (3.5,-7.5) node[]{\Huge $\lrc$};
    \draw (5.5,-7.5) node[]{\Huge $\llc$};
    \draw (7.5,-7.5) node[]{\Huge $\lrc$};
  \end{scope}
      \draw[opacity=0.3] (.7,2.1,0) node{$\ell = 1\ (\cdot,\uo)$};
      \draw (1.7,2.1,0) node{$\ell = 2\ (q_0,\uo)$};
      \draw[opacity=0.3] (2.7,2.1,0) node{$\ell = 1\ (\cdot,\uc)$}; 
      \draw[opacity=0.3] (3.7,2.1,0) node{$\ell = 3\ (\cdot,\uo)$}; 
      \draw[opacity=0.3] (4.7,2.1,0) node{$\ell = 1\ (\cdot,\uo)$};
      \draw[canvas is xz plane at y=-1.5, gray, ->] (-4,0) -- ++(0,1) node[at end, above, gray] {$k$};
      \draw[canvas is xz plane at y=-1.5, gray, ->] (-4,0) -- ++(1,0) node[at end,below, gray] {$i$};
      \draw[canvas is yz plane at x=-4, gray, ->] (-1.5,0) -- ++(-1,0) node[at end,above, gray] {$j$};
\end{tikzpicture}

  \caption{A configuration of $X_3^\uo$. The horizontal plane contains a configuration $x\in X_T^2$, and the slice of level 2 marked by $(q_0,\uo)$ contains a configuration of $D$ ``synchronised'' with the squares of $x$.\label{fig:f3o}}
\end{figure}
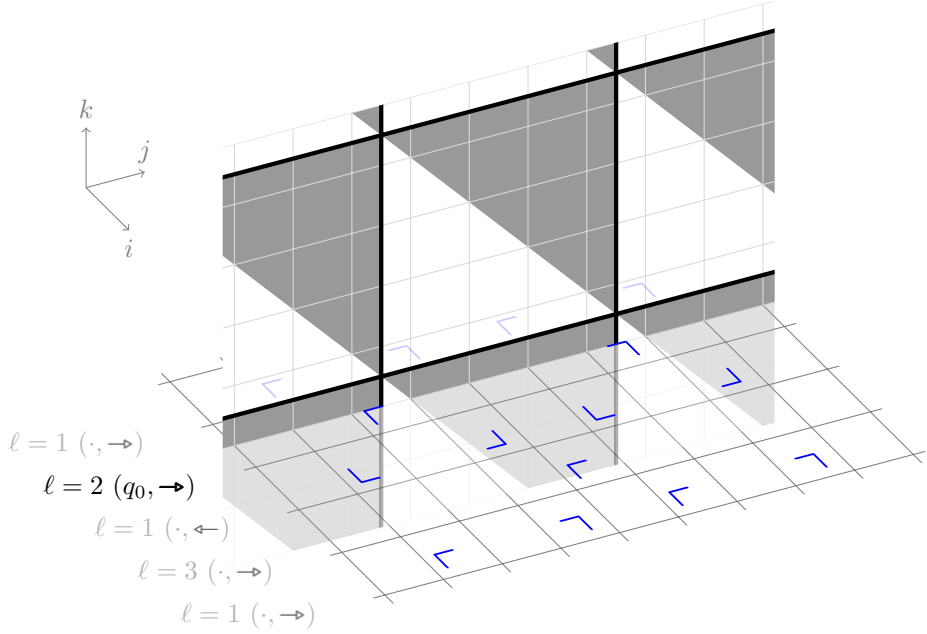

See \Cref{fig:f3o} for a visual help. Formally, if $\Sigma_X$ and $\Sigma_D$ are the alphabets of $X_2$ and $D$:
\[X_3^\uo = 
  \left\{ 
  x \in {\left(\Sigma_X \times \Sigma_D \right)}^{\Z^3} :\: 
  \begin{array}{l} \exists y \in X_2,\ \forall i,j, k \in \Z,\ \pi_{1,2,3,4}(x_{i,j,k}) = y_{i,j} \\
   \forall i\in\Z,\ \exists d^i \in D,\ \forall j,k \in \Z,\ \pi_5(x_{i,j,k}) = d^i_{j,k} \\
   \forall i,j,k \in \Z,\begin{array}{l} \pi_{1,2}(x_{i,j,k}) \neq (q_0,\uo) \implies \pi_5(x_{i,j,k}) = \square \\
   \pi_5(x_{i,j,k}) \in\{ \begin{tikzpicture}[baseline=0.3em]
    \fill[lipicsBulletGray] (0,0) rectangle (0.4,0.4);
    \fill[white] (0,0.2) -- (0.2,0.2) -- (0,0.4) -- cycle;
    \fill[white] (0.2,0.2) rectangle (0.4,0.4);
    \fill[white] (0.2,0.2) -- (0.4,0) -- (0.2,0) -- cycle;
    \draw[very thick] (0,0.2) -- (0.4,0.2);
    \draw[very thick] (0.2,0) -- (0.2,0.4);
    \draw[lipicsBulletGray] (0,0) rectangle (.4, .4);
  \end{tikzpicture},\begin{tikzpicture}[baseline=0.3em]
    \fill[white] (0,0) rectangle (0.2,0.4);
    \fill[lipicsBulletGray] (0.2,0) rectangle (0.4,0.4);
    \draw[very thick] (0.2,0) -- (0.2,0.4);
     \draw[lipicsBulletGray] (0,0) rectangle (.4, .4);
  \end{tikzpicture}\}
  \iff \pi_3(x_{i,j,k}) \in \{ \ulc,\llc\} \end{array}
\end{array}\right\}
\]

As there exists at most a single square of infinite level in $X_T^2$, there exists at most a single slice $(i,\cdot,\cdot)$ of infinite level in $X_3^{\uo}$.

On every slice marked by $\uo$, the configuration of $D$ breaks all periods smaller than its squares, and the size of the squares is controlled by the level of the slice in $T_\mathcal{M}$. Therefore:

\begin{claim}\label{claim:x3-period-breaker}
A configuration of $X_3^{\uo}$:
\begin{enumerate}
\item breaks every periodicity vector $(n,\cdot,\cdot)$ and $(\cdot,n,\cdot)$ for $n \geq 1$.
\item for every slice $(i,\cdot,\cdot)$ containing $(q_0,\uo)$ on the first two layers and corresponding to a column of level $\ell$ in $X_T^2$, Layer 5 breaks every vector $(\cdot,\cdot,n)$ for $1 \leq n < 2^\ell$.
\end{enumerate}
\end{claim}
\begin{claimproof}
\begin{enumerate}
\item  $X_2$ is aperiodic because $X_T^2$ is also aperiodic, so all vectors $(n,\cdot,\cdot)$ and $(\cdot,n,\cdot)$ are broken for $n \geq 1$.
\item  For a slice $(i,\cdot,\cdot)$ of level $\ell$, the distance along $(0,1,0)$ between two consecutive lines in Layer 3 (ie. in $X_T^2$) is exactly $2^\ell$ (see. \Cref{sec:2d-toeplitz}, \hyperlink{points-2d-toeplitz}{point 3} in the list of properties of configurations in $X_T^2$). So, Layer 5 breaks every smaller period in this direction.\claimqedhere
\end{enumerate}\end{claimproof}

\subsection{\texorpdfstring{$\Z^4$}{Z\^{}4} SFT \texorpdfstring{$X_4$}{X\_4} and proof of the reduction}\label{sec:4d-final}

\paragraph*{Creation of $X_4$}
Similarly to $Y_3$ (\Cref{sec:sofic-y3}), we build $X_4$ by ``fusing together'' $X_3^\uo$ and $X_3^\uc$. Formally:
\begin{align*}
  X_4 = \{x \in {\left(\Sigma_X \times \Sigma_D \times \Sigma_D \right)}^{\Z^4} :\:& \exists x^\uo \in X_3^\uo, \exists x^\uc \in X_3^\uc,\\
  \forall i,j,k,l \in \Z,\; \pi_{1,2,3,4,5}&(x_{i,j,k,l}) = x^\uo_{i,j,l} \text{ and } \pi_{1,2,3,4,6}(x_{i,j,k,l}) = x^\uc_{i,j,k} \}
\end{align*}
\begin{claim}\label{claim:x4-sft}
  $X_4$ is an SFT.
\end{claim}
\begin{claimproof}
  Both $X_3^\uo$ and $X_3^\uc$ are SFTs, since $X_2$ is an SFT.
\end{claimproof}

\paragraph*{Reduction \textbf{RS} $\leq$ \textbf{AD}}

\begin{lemma}\label{claim:x4-aperiodicity}
  There exists an aperiodic configuration in $X_4$ if and only if there exists a run of $\mathcal{M}$ in which $q_0$ occurs infinitely often.
\end{lemma}

\begin{proof}\hypertarget{prooflem:x4-aperiodicity}
This is the same proof as for \Cref{claim:y3-aperiodicity}, except that vectors along the first two dimensions are broken by the Toeplitz structure on layer 3. Otherwise, Layers 5 and 6 break every vector $(0,0,\cdot,\cdot)$ if and only if the run visits $q_0$ infinitely often.
\end{proof}

This concludes the proof of \Cref{th:4d-sft}.

\section{Complexity and aperiodic configurations}\label{sec:low-complexity}

The \emph{complexity function} of a $\Z^d$ subshift $X$ is $N_X(n) = \# \{ w \in \Sigma^{{\llbracket 0, n-1 \rrbracket}^d} : \exists x \in X, w \sqsubseteq x \}$. In this section, we see that subshifts of high complexity are expected to have aperiodic configurations. 

\begin{definition}[Dimensional entropy]\label{def:entropies}
Define $h_k(X)$, \emph{the entropy of dimension $k$}, as:
\begin{align*}
 h_k(X) & = \limsup_{n \to +\infty} \frac{\log N_X(n)}{n^k} \in [0,+\infty]
\end{align*}
\end{definition}

\cite[Theorem 10]{2018-GHV} proves that a $\Z^2$ subshift $X$ with no aperiodic configurations is almost topologically conjugated to (i.e. ``nearly behaves as'') a $\Z$ subshift of the same type. In particular, $h_1(X)=+\infty$ is only possible when $X$ contains an aperiodic configuration.

\cite[Corollary 13]{2009-Hochman} entails that $h_d(X)>0$ for a $\Z^d$ SFT implies the existence of aperiodic configurations in any dimension. We improve this result as follows:

\begin{proposition}\label{th:low-complexity-and-eac}
Let $X$ be a $\Z^d$ SFT or sofic subshift. If $h_{d-1}(X) = +\infty$, then there exists an aperiodic configuration in $X$.
\end{proposition}

\begin{proof}
Let $X' \subseteq {\Sigma'}^{\Z^d}$ be a SFT cover of $X$: $\pi(X') = X$ for $\pi$ a letter-to-letter projection. W.l.o.g., assume $X'$ is defined by adjacency constraints. Consider all patterns on $\llbracket 0, n-1 \rrbracket^d$ admissible in $X'$ : they have exactly $N_X(n)$ different projections by $\pi$, but the number $N^b_{X'}(n)$ of different patterns on the boundary of $\llbracket 0, n-1 \rrbracket^d$ is at most ${\Sigma'}^{2d n^{d-1}}$. As $\log N_X(n)/n^{d-1}$ is unbounded, there exists some $n$ such that:
  \[ 2d (\log \Sigma') < \frac{\log N_X(n)}{n^{d-1}}\quad \iff \quad{\Sigma'}^{2d n^{d-1}}  < N_X(n) \quad\implies\quad  N^b_{X'}(n) < N_X(n)\]
By the pigeonhole principle, there exists a pattern $b$ on the boundary admissible in $X'$ that can be extended on the cube in two different admissible patterns $b^+, b^-$ such that $\pi(b^+) \neq \pi(b^-)$. Consider a configuration $x' \in X'$ in which $b^+$ appears. If $\pi(x')$ is not already aperiodic, swapping $b^+$ for $b^-$ in $x'$ at a some arbitrary position leads to a configuration $\pi(x'') \in X$ which is aperiodic.
\end{proof}

\Cref{th:low-complexity-and-eac} is tight: the $\Z^{d}$-lift of a $\Z^{d-1}$ SFT ($d \geq 3$) is periodic by definition, and can have arbitrarily high entropy of dimension~$(d-1)$. We conjecture that Proposition~\ref{th:low-complexity-and-eac} holds for all subshifts even in dimension $d>2$.

\Cref{th:low-complexity-and-eac} shows that \textbf{AD} is a problem that is only relevant for low complexity subshifts, which is where its full computational complexity ``lies''. Indeed, the problem of deciding whether $h_{d-1}(X) = +\infty$ is~$\Pi_3^0$ (it is equivalent to $\forall k\ \exists n\ N_X(n)>k$, and $N_X(n)$ is a $\Pi_1^0$ integer), which is much easier than the $\Sigma_1^1$-completeness of \textbf{AD} on $\Z^3$ sofic subshifts.

\section{Open problems}

The main remaining question is, of course, the case of $\Z^3$ SFTs. The method we developed above to prove $\Sigma_1^1$-completeness in the case of $\Z^3$ sofic subshifts cannot be applied. Indeed, embedding computations in an SFT requires at least two aperiodic dimensions; and we need two other dimensions (\hyperlink{par:infinite-level-aper}{because of the positions of level $\infty$}) which can be periodic or aperiodic.

We conjecture that aperiodic configurations in $\Z^3$ SFTs behave similarly as in $\Z^2$ subshifts: each $\Z^3$ SFT containing aperiodic configurations seems to have ``centers'' of aperiodicity, i.e. concentric zones in which periods are broken. The distance from the center might depend on the length of the vector, $|\Sigma|$ and the size of the largest forbidden pattern.

However, there seem to be important differences. First, for $\Z^3$-SFT, not all aperiodic configurations have a center of aperiodicity in their orbit closure: this center may be in an unrelated aperiodic configuration. Second, results of \cite{2018-GHV} are valid for all $\Z^2$ subshifts, but our conjecture must be specific to $\Z^3$ SFTs, and a proof requires SFT-specific techniques. \medskip

Considering subshifts on more general groups (other than $\Z^d$), there is an active research theme looking for conditions on groups which make the Domino problem undecidable. In this context, we would like to obtain conditions that make \textbf{AD} $\Pi_1^0$- or $\Sigma_1^1$-complete.

\bibliographystyle{plainurl}
\bibliography{bibliography}
\end{document}